\documentclass[12pt,journal,final,onecolumn]{IEEEtran}

\usepackage[dvips]{graphicx}
\usepackage{epsfig}
\usepackage[cmex10]{amsmath}
\usepackage{amssymb}
\usepackage{amsthm}
\usepackage{amsfonts}
\usepackage{bm}
\usepackage{dsfont}
\usepackage{color}
\usepackage[mathscr]{eucal}
\usepackage{cite}
\usepackage{array}
\usepackage{fixltx2e}
\usepackage{bookmark}

\graphicspath{{figs/}}

\input{niesen_defs.def}

\begin{document}

\title{Caching in Wireless Networks} 

\author{Urs~Niesen, Devavrat~Shah, Gregory W. Wornell%
\thanks{U. Niesen was with the Laboratory for Information and Decision
Systems and with the Research Laboratory of Electronics, Department of
EECS, MIT. He is now with the Mathematics of Networks and Communications
Research Department, Bell Labs, Alcatel-Lucent.
Email: {\tt urs.niesen@alcatel-lucent.com}}%
\thanks{D. Shah is with the Laboratory for Information and Decision
Systems, Department of EECS, MIT. Email: {\tt devavrat@mit.edu}}%
\thanks{G. W. Wornell is with the Research Laboratory of Electronics,
Department of EECS, MIT. Email: {\tt gww@mit.edu}}%
\thanks{This work was supported, in parts, by DARPA under Grant No.
18870740-37362-C, AFOSR under Grant No. FA9550-09-1-0317, 
and by NSF under Grant No. CCF-0635191.}}

\maketitle

\begin{abstract}
    We consider the problem of delivering content cached in a wireless
    network of $n$ nodes randomly located on a square of area $n$. The
    network performance is described by the $2^n\times n$-dimensional
    caching capacity region of the wireless network. We provide an inner
    bound on this caching capacity region, and, in the high path-loss
    regime, a matching (in the scaling sense) outer bound. For large
    path-loss exponent, this provides an information-theoretic scaling
    characterization of the entire caching capacity region. The proposed
    communication scheme achieving the inner bound shows that the
    problems of cache selection and channel coding can be solved
    separately without loss of order-optimality.  On the other hand, our
    results show that the common architecture of nearest-neighbor cache
    selection can be arbitrarily bad, implying that cache selection and
    load balancing need to be performed jointly.
\end{abstract}

\section{Introduction}
\label{sec:intro}

Wireless networks are an attractive communication architecture in many
applications as they require only minimal fixed infrastructure.  While
unicast and multicast traffic in wireless networks has been widely
studied, the influence of caches on the network performance has received
considerably less attention. Nevertheless, the ability to replicate data
at several places in the network is likely to significantly increase
supportable rates. In this paper, we consider the problem of
characterizing achievable rates with caching in large wireless networks.

In a rather general form, this problem can be formulated as follows.
Consider a wireless network with $n$ nodes, and assume a node $w$ in the
network requests a message available at the set of caches $U$ (a subset
of the nodes) at a certain rate $\lambdaca_{U,w}$. The collection of all
$\lambdaca_{U,w}$ can be represented as a caching traffic
matrix $\lambdaca\in\Rp^{2^{n}\times n}$. The question is then to
characterize the set of achievable caching traffic matrices
$\Lambdaca(n)\subset\Rp^{2^{n}\times n}$.

\subsection{Related Work}
\label{sec:prior}

Several aspects of caching in wireless networks have been investigated
in prior work. In the computer science literature, the wireless network
is usually modeled as a graph induced by the geometry of the node
placement. This is tantamount to making a protocol model assumption (as
proposed in \cite{gup}) about the communication scheme used. By
definition, such an approach assumes separation of source and channel
coding. The quantity of interest involves the distance from each node to
the closest cache that holds the requested message.  The problem of
optimal cache location for multicasting from a single source has been
investigated in \cite{nug, bha}. Optimal caching densities under uniform
random demand have been considered in \cite{jin, ko}. Several cache
replacement strategies are proposed, for example, in \cite{yin}.

To the best of our knowledge, caching has not been directly considered
in the information theory literature. However, the more general problem
of transmitting correlated sources over a network has been studied.
Caching is a special case of this problem, in which sources are either
independent or identical.  While for a single point-to-point channel
separation of source and channel coding was shown to be optimal by
Shannon \cite{sha}, the work by Cover, El Gamal, and Salehi \cite{cov}
established that separation is strictly suboptimal for the transmission
of correlated sources over a multiple access channel.  Hence, even for
simple networks, source and channel coding have to be considered
jointly. We note that for some special cases separate source and channel
coding is optimal, for example for transmitting arbitrarily correlated
sources over a network consisting of independent point-to-point links
\cite{han, bar, tia}. The general problem of joint source-channel coding
for noisy networks is unsolved.

Finally, it is worth mentioning the problem of transmitting unicast
traffic over a wireless network, which is a special case of the caching
problem with each message being available at only a single cache. This
problem has been widely studied. Approximate characterizations of the
unicast capacity region of large wireless networks (also known as
scaling laws) were derived, for example, in
\cite{gup,xie,jov,lev,xue,xie2,fra,gup2,ozg,nie,fra2,nie2}.

\subsection{Summary of Results}

We consider the general caching problem from an information-theoretic
point of view. Compared to the prior work mentioned in the last section,
there are several key differences. First, we do not make a protocol
channel model assumption, and instead allow the use of arbitrary
communication protocols over the wireless network including joint
source-channel coding. Second, we allow for general traffic demands,
i.e., arbitrary number of caches, and arbitrary demands at each
destination. Third, we do not impose that each destination requests the
desired message from only the closest cache, nor do we impose that the
entire message be requested from the same cache. Rather, we allow
parts of the same message to be requested from different caches.

We present a communication scheme for the caching problem, yielding an
inner bound on the caching capacity region $\Lambdaca(n)$. This
communication scheme performs separate source and channel coding. For
large values of path-loss exponent, we provide a matching (in the
scaling sense) outer bound, proving the approximate optimality of our
proposed scheme for large values of $n$. Together, this provides a
scaling description of the entire caching capacity region of the
wireless network in the large path-loss regime. This result further
implies that for caching traffic the loss due to source-channel
separation is small (again in the scaling sense) in the large path-loss
regime. Since caching traffic is a special case of correlated sources,
in which two sources are either identical or independent, this result is
a step towards understanding the loss incurred due to source-channel
separation for the transmission of arbitrarily correlated sources.

\subsection{Organization}

The remainder of this paper is organized as follows. Section
\ref{sec:model} introduces the channel model and notation.  Section
\ref{sec:main} presents the main results of the paper. Section
\ref{sec:proofs} analyzes the proposed communication scheme and
establishes its optimality (up to scaling) for large path-loss exponent.
Section \ref{sec:conclusions} contains concluding remarks.

\section{Network Model and Notation}
\label{sec:model}

Consider a square of area $n$, denoted by
\begin{equation*}
    A(n) \defeq [0,\sqrt{n}]^2.
\end{equation*}
Let $V(n)\subset A(n)$ be a set of $\abs{V(n)} = n$ nodes placed
independently and uniformly at random on $A(n)$. We assume the following
complex baseband-equivalent channel model. The received signal at node
$v$ and time $t$ is
\begin{equation*}
    y_v[t] \defeq \sum_{u\in V(n)\setminus\{v\}}h_{u,v}[t]x_u[t]+z_v[t]
\end{equation*}
for all $v\in V(n), t\in\mbb{N}$, and where $x_u[t]$ is the
channel input at node $u$ at time $t$. Here
$(z_v[t])_{v,t}$ are independent and identically distributed (i.i.d.)
circularly symmetric complex Gaussian random variables with mean $0$ and
variance $1$, and
\begin{equation*}
    h_{u,v}[t] \defeq r_{u,v}^{-\alpha/2}\exp(\sqrt{-1}\theta_{u,v}[t]),
\end{equation*}
for \emph{path-loss exponent} $\alpha>2$, and where $r_{u,v}$ is the
Euclidean distance between $u$ and $v$. Due to physical constraints, the
path-loss exponent $\alpha$ satisfies $\alpha \geq 2$; we adopt the
slightly stronger assumption $\alpha > 2$ because it simplifies the
statements and derivations of some of the results. The phase terms
$(\theta_{u,v}[t])_{u,v}$ are assumed to be i.i.d.\  with uniform
distribution on $[0,2\pi)$.\footnote{It is worth pointing out that the
i.i.d.\ assumption on the phase terms has to be made with some care.  In
particular, it is shown in \cite{fra2, lee10, ozg10b} that this
assumption is valid only if the wavelength of the carrier frequency is
less than $\card{A(n)}^{1/2}/n$. For a wide range of scenarios this is
the case, and we assume throughout this paper that this assumption
holds.} We either assume that $(\theta_{u,v}[t])_{t}$ is stationary
and ergodic as a function of $t$, which is called \emph{fast fading} in
the following, or we assume that $(\theta_{u,v}[t])_{t}$ is constant
as a function of $t$, which is called \emph{slow fading} in the
following. In either case, we assume full channel state information
(CSI) is available at all nodes, i.e., each node knows all
$(h_{u,v}[t])_{u,v}$ at time $t$.\footnote{We make the full CSI
assumption in all the converse results in this paper.  Achievability can
be shown to hold under weaker assumptions on the availability of CSI. In
particular, for $\alpha\geq 3$, no CSI is necessary, and for
$\alpha\in(2,3)$, a $2$-bit quantization of the channel state
$(\theta_{u,v}[t])_{u,v}$ available at all nodes at time $t$ is
sufficient.} We also impose an average unit power constraint on the
channel inputs $(x_u[t])_{t}$ for every node $u\in V(n)$.

A \emph{caching traffic matrix} is an element
$\lambdaca\in\Rp^{2^{n}\times n}$. Consider $w\in V(n)$ and $U\subset
V(n)$. Assume a message that is requested at destination node $w$ is
available at all of the caches $U$.  $\lambdaca_{U,w}$ denotes then the
rate at which node $w$ requests the message from the caches
$U$.\footnote{Note that several rates $\lambdaca_{U,w}$ are trivial. For
example for pairs $(U,w)$ with $w \in U$, or for pairs $(U,w)$ with
$U=\emptyset$. We allow these trivial choices for notational
convenience. For $(U,w)$ such that $w\in U$, the results will show that
$\lambdaca_{U,w}=\infty$ is achievable; for $U=\emptyset$, they will
show that only $\lambdaca_{U,w}=0$ is achievable, as would be expected.}
Note that we do not impose that any particular cache $u\in U$ provides
$w$ with the desired message, rather multiple nodes in $U$ could
provide parts of the message.  Note also that $\lambdaca_{U,w}$ and
$\lambdaca_{\tilde{U},w}$ could both be strictly positive for $U\neq
\tilde{U}$, i.e., the same destination could request more than one
message from different collection of caches. We assume that messages for
different $(U,w)$ pairs are independent. The \emph{caching capacity
region} $\Lambdaca(n)$ of the wireless network $V(n)$ is the closure of
the set of all achievable caching traffic matrices
$\lambdaca\in\Rp^{2^{n}\times n}$.

\begin{example}
    Consider $V(n)=\{v_i\}_{i=1}^4$ with $n=4$. Assume that $v_1$
    requests a message $m_{\{v_3,v_4\},v_1}$ available at the caches
    $v_3$, and $v_4$ at rate $1$ bit per channel use, and an independent
    message $m_{\{v_3\},v_1}$ available only at $v_3$ at a rate of $2$
    bits per channel use. Node $v_2$ requests a message
    $m_{\{v_3,v_4\},v_2}$ available at the caches $v_3$ and $v_4$ at a
    rate of $4$ bits per channel use. The messages
    $m_{\{v_3,v_4\},v_1}$, $m_{\{v_3\},v_1}$, and $m_{\{v_3,v_4\},v_2}$
    are assumed to be independent. This
    traffic pattern can be described by a caching traffic matrix
    $\lambdaca\in\Rp^{16\times 4}$ with $\lambdaca_{\{v_3,v_4\},v_1}=1$,
    $\lambdaca_{\{v_3\},v_1}=2$, $\lambdaca_{\{v_3,v_4\},v_2}=4$, and
    $\lambdaca_{U,w}=0$ otherwise.  Note that in this example node $v_1$
    is destination for two (independent) caching messages, and node
    $v_3$ and $v_4$ serve as caches for more than one message (but these
    messages are again assumed independent). 
\end{example}

To simplify notation, we assume when necessary that large reals are
integers and omit $\ceil{\cdot}$ and $\floor{\cdot}$ operators. For the
same reason, we suppress dependence on $n$ within proofs whenever this
dependence is clear from the context. We use bold font to denote
matrices whenever the matrix structure is of importance. We use the
$\dagger$ symbol to denote the conjugate transpose of a matrix. Finally,
$\log$ and $\ln$ represent the logarithms with respect to base $2$ and
$e$, respectively.

\section{Main Results}
\label{sec:main}

The main results of this paper are an achievable scheme and an outer
bound for the caching capacity region $\Lambdaca(n)$. Section
\ref{sec:main_region} describes a construction used in Section
\ref{sec:main_inner} to establish an inner bound for $\Lambdaca(n)$.
The communication scheme achieving this inner bound respects
source-channel separation and is valid for any value of path-loss
exponent $\alpha > 2$. In Section \ref{sec:main_outer}, we provide an
outer bound that matches (in the scaling sense) the inner bound for
large values of path-loss exponent $\alpha > 6$. This leads to an
approximate characterization of $\Lambdaca(n)$ for $\alpha > 6$. This
characterization is given in terms of a linear program and is hence
computationally tractable as is discussed in Section
\ref{sec:main_comput}. The communication architecture achieving the
inner bound on the caching capacity region is presented in Section
\ref{sec:main_scheme}. Various example scenarios are presented in
Section \ref{sec:main_examples}.

\subsection{Tree Graph and Linear Program}
\label{sec:main_region}

We describe the construction of a capacitated tree graph induced by the
wireless network and a corresponding linear program. These will be
needed for the communication scheme achieving the inner bound. This tree
graph construction was introduced first in \cite{nie2}. 

Partition the square $A(n)$ into $4^\ell$ subsquares
$\{A_{\ell,i}(n)\}_{i=1}^{4^{\ell}}$ of sidelength $2^{-\ell}\sqrt{n}$,
and let $V_{\ell,i}(n)$ be the nodes in $A_{\ell,i}(n)$. The integer
parameter $\ell$ varies between $0$ and 
\begin{equation*}
    L(n) \defeq \frac{1}{2}\log(n)\big(1-\log^{-1/2}(n)\big).
\end{equation*}
The partitions at various levels $\ell$ form a dyadic decomposition of
$A(n)$ as illustrated in Fig. \ref{fig:grid}. The choice of $L(n)$ is
made such that with high probability the number of nodes in each set
$V_{L(n),i}$ at the finest grid level is growing to infinity, but not
too quickly. See \cite{nie2} for a detailed discussion.
\begin{figure}[tbp]
    \begin{center}
        \scalebox{0.7}{
        \input{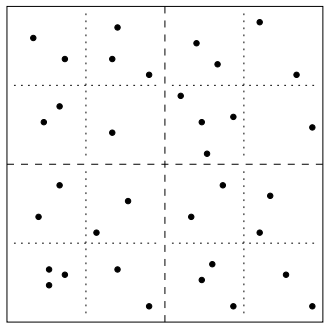}
        }
    \end{center}

    \caption{Subsquares $\{A_{\ell,i}(n)\}$ with $0\leq\ell\leq 2$, i.e.,
    with $L(n)=2$.  The subsquare at level $\ell=0$ is the area $A(n)$
    itself. The subsquares at level $\ell=1$ are indicated by dashed
    lines, the subsquares at level $\ell=2$ by dotted lines. Assume for
    the sake of example that the subsquares are numbered from left to
    right and then from bottom to top (the precise order of numbering is
    immaterial).  Then $V_{0,1}(n)$ are all the nodes $V(n)$,
    $V_{1,1}(n)$ are the nine nodes in the lower left corner (delineated
    by dashed lines), and $V_{2,1}(n)$ are the three nodes in the lower
    left corner (delineated by dotted lines).}

    \label{fig:grid}
\end{figure}

We now construct an undirected, capacitated tree graph $G=(V_G,E_G)$
as depicted in Fig. \ref{fig:grid_graph}.
\begin{figure}[tbp]
    \begin{center}
        \scalebox{0.89}{
        \input{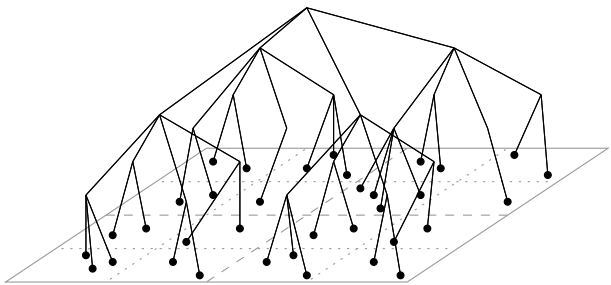}
        }
    \end{center}

    \caption{Construction of the tree graph $G$. We consider the same
    nodes as in Fig. \ref{fig:grid} with $L(n) = 2$. The leaves of $G$
    are the nodes $V(n)$ of the wireless network. They are always at
    level $\ell=L(n)+1$ (i.e., $3$ in this example). At level
    $0\leq\ell\leq L(n)$ in $G$, there are $4^\ell$ nodes. The tree
    structure is induced by the decomposition of $V(n)$ into subsquares
    $\{V_{\ell,i}(n)\}_{\ell,i}$, delineated by dashed and dotted lines.
    Level $0$ contains the root node of $G$.
    }

    \label{fig:grid_graph}
\end{figure}
The vertex set $V_G$ of $G$ consists of the nodes $V(n)$ in the wireless
network plus some additional nodes. The tree $G$ has $L(n)+2$ levels
numbered $0$ to $L(n)+1$: the root node is at level $0$ and leaf nodes
are at level $L(n) + 1$. The leaf nodes of $G$ are the $n$ nodes $V(n)$
in the wireless network. The nodes of $G$ at level $\ell$ with $1\leq
\ell \leq L(n)$ are elements of $V_G\setminus V(n)$ and correspond to
subsets $\{V_{\ell, i}(n)\}_{i=1}^{4^\ell}$ of the nodes $V(n)$ in the
wireless network. The root node of $G$ at level $0$ corresponds to all
the nodes $V(n)$ in the wireless network. A child node at level $\ell+1$
is connected to a parent node at level $\ell$ as follows. For $\ell =
L(n)$, a node $v$ at level $L(n)+1$ (which is a leaf node of $G$ and
hence also an element of the nodes $V(n) \subset V_G$ in the wireless
network) is connected to the node in $G$ corresponding to $V_{L(n),
i}(n)$ if $v$ belongs to $V_{L(n), i}(n)$. For $0\leq \ell < L(n)$, a
node in $G$ at level $\ell + 1$ corresponding to $V_{\ell+1, i}(n)$ is
connected to the node in $G$ corresponding to $V_{\ell, j}(n)$ if
$V_{\ell+1, i}(n)  \subset V_{\ell, j}(n)$. 

Note that through this construction, each set $V_{\ell,i}(n)$ for
$\ell\in\{0,\ldots,L(n)\}$, $i\in 4^\ell$ is represented by exactly one
internal node in $G$. Thus, the cardinality of $V_G$ is
\begin{align}
    \label{eq:cardvg}
    \card{V_G} 
    & = \card{V(n)} +  \sum_{\ell=0}^{L(n)} 4^\ell \nonumber\\
    & = n + \frac{1}{3}\Big(4^{L(n)+1}-1\Big) \nonumber\\
    & \leq 2n.
\end{align}

We assign to each edge $e\in E_G$ at level $\ell$ in $G$ (i.e., between
nodes at levels $\ell$ and $\ell-1$) a capacity
\begin{equation*}
    c_e \defeq
    \begin{cases}
        (4^{-\ell}n)^{2-\min\{3,\alpha\}/2} & \text{if $1\leq\ell\leq L(n)$,} \\
        1 & \text{if $\ell = L(n)+1$.}
    \end{cases}
\end{equation*}
With slight abuse of notation, we let for $(u,v)=e\in E_G$
\begin{equation*}
    c_{u,v} \defeq c_e.
\end{equation*} 

The capacity $c_e$ associated with an edge $e=(u,v)$ is to be
interpreted as follows. Recall that the nodes $u$ and $v$ in $G$
correspond to a subset of nodes in the wireless network. Let nodes $u$
and $v$ in $G$ be at levels $\ell-1$ and $\ell$ with $1\leq \ell\leq
L(n)$.  The corresponding subsets $V_{\ell-1,i}(n)$ and $V_{\ell,j}(n)$
(for some $i$ and $j$) have approximately $4^{-\ell+1}n$ and
$4^{-\ell}n$ nodes with high probability. Assume we could cooperatively
communicate from $V_{\ell-1,i}(n)$ to the nodes $V_{\ell,j}(n)$ in the
wireless network.  This results in a large multiple-input
multiple-output (MIMO) channel with approximately $4^{-\ell+1}n$
transmit and $\tfrac{3}{4}4^{-\ell}n$ receive antennas. The capacity of
this MIMO channel can be evaluated to be approximately
$(4^{-\ell}n)^{2-\min\{3,\alpha\}/2}$. Similarly, for a node $u$ at
level $L(n)+1$, the capacity from $u$ to the set $V_{L(n),i}$ it is
contained in is approximately equal to one. Thus, we see that the edge
capacity $c_e$ is approximately equal to the MIMO capacity between the
subsets in the wireless network corresponding to the nodes in $G$
connected by $e$. 

Recall that the leaf nodes of $G$ are equal to the nodes $V(n)$ in of
the wireless network. Hence, any caching traffic matrix $\lambdaca \in
\Rp^{2^{n}\times n}$ for the wireless network is also a valid traffic
matrix between leave nodes of $G$. Assume the leaf nodes of $G$ request
messages according to the caching traffic matrix $\lambdaca$.
Specifically, we wish to route data from caches in $U \subset V(n)$ to a
node $w \in V(n)$ over $G$ at rate $\lambdaca_{U,w}$. We say that
$\lambdaca$ is \emph{supportable on $G$} if this is possible. Let
$\gLambdaca(n)$ denote the collection of all caching traffic matrices
$\lambdaca \in \Rp^{2^{n}\times n}$ that are supportable on $G$. It can
be verified that $\gLambdaca(n)$ is a closed convex set containing the
origin.  

Given the tree structure of $G$, there is unique path connecting any two
of its nodes. The only way to satisfy the rate demand $\lambdaca_{U, w}$
by routing is to split it amongst different $(u,w)$ pairs with $u\in U$.
Specifically, let $P_{U,w}$ denote the set of $|U|$ unique paths in $G$
between nodes of $U$ and $w$. For a path $p\in P_{U,w}$ between $u\in U$
and $w$, let $f_{p,U}$ be the rate at which demand is routed from node
$u\in U$ to $w$ along path $p$ for request $(U,w)$. A caching traffic
matrix $\lambdaca$ is supportable on the capacitated graph $G$ if and
only if for each of the $2^{n}\times n$ pairs $(U,w)$ there exists a
decomposition 
\begin{equation*}
    \lambdaca_{U,w} = \sum_{p \in P_{U,w}} f_{p,U}
\end{equation*}
so that the resulting load on each edge of $G$ is no more than its
capacity. Formally, consider the following linear program
\begin{equation}\label{eq:computation}
    \begin{array}{lr@{\,}ll}
        \text{max} & \phi & & \\
        \text{s.t.} & {\displaystyle\sum_{p\in P_{U,w}}} f_{p,U} 
        & \geq \phi \lambdaca_{U,w} & \!\!\!\forall U\subset V,w\in V , \\
        & \!\!\!\!\!\!\!\!\!\!\!\!\!\!\!{\displaystyle\sum_{U\subset V} \sum_{w\in V}
        \sum_{\substack{p\in P_{U,w}:\\ e\in p}}} f_{p,U} 
        & \leq c_e & \!\!\!\forall e \in E_{G}, \\
        & f_{p,U} 
        & \geq  0 & \!\!\!\forall U\subset V, w\in V, p\in P_{U,w},
    \end{array}
\end{equation}
with $V=V(n)$, and where the maximization is over the variables $\phi$
and $f_{p,U}$. Denote the maximum value of $\phi$ by $\phi(\lambdaca)$.
The caching traffic matrix $\lambdaca$ is supportable on the graph $G$,
if and only if $\phi(\lambdaca) \geq 1$. 

Note that for any $\lambdaca \in \Rp^{2^{n}\times n}$, the caching
traffic matrix $\phi(\lambdaca)\lambda$ is supportable on $G$, i.e.,
$\phi(\lambdaca)\lambdaca\in \gLambdaca(n)$. Thus,
\begin{equation*}
    \phi(\lambdaca) 
    = \max\Bigl\{ \phi \geq 0 : \phi \lambdaca \in \gLambdaca(n)\Bigr\}.
\end{equation*}
In words, $\phi(\lambda)$ is the largest multiple such that the scaled
traffic matrix $\phi(\lambdaca)\lambdaca$ is supportable on $G$. 
Since $\gLambdaca(n)$ is a closed convex set
containing the origin, knowledge of $\phi(\lambdaca)$ for all
$\lambdaca\in\Rp^{2^n\times n}$ completely specifies $\gLambdaca(n)$. We
can think of  $\phi(\lambdaca)$, evaluated for all $\lambdaca$, as an
equivalent description of the region $\gLambdaca(n)$.

\subsection{Inner Bound}
\label{sec:main_inner}

The first result provides an inner bound for the caching capacity region
$\Lambdaca(n)$ in terms of the set $\gLambdaca(n)$ of supportable
caching traffic matrices over the graph $G$. This result is valid for
all $\alpha > 2$, i.e., for all values of the path-loss exponent
$\alpha$ of interest (excluding the boundary point $\alpha=2$ as
discussed in Section \ref{sec:model}).

For $\lambdaca \in \Rp^{2^{n}\times n}$, define 
\begin{equation*}
    \rhoca(\lambdaca)
    \defeq \max\Bigl\{ \rho \geq 0 : \rho \lambdaca \in \Lambdaca(n)\Bigr\}.
\end{equation*}
In words, $\rhoca(\lambdaca)$ is the largest multiple such that the scaled
traffic matrix $\rhoca(\lambdaca)\lambda$ is achievable over the
wireless network. The caching capacity region $\Lambdaca(n)$ is a closed
convex set containing the origin, and hence $\rhoca(\lambdaca)$ 
is an equivalent description of $\Lambdaca(n)$.

\begin{theorem}
    \label{thm:caching1}
    Under either fast or slow fading, for any $\alpha>2$, there exists
    $b_1(n) \geq n^{-o(1)}$ such that 
    \begin{equation*}
        \rhoca(\lambdaca) \geq b_1(n)\phi(\lambdaca) 
    \end{equation*}
    for all $\lambdaca \in \Rp^{2^{n}\times n}$ 
    with probability $1-o(1)$ as $n\to\infty$.
\end{theorem}

The proof of Theorem \ref{thm:caching1} is provided in Section
\ref{sec:inner_proof}.  We point out that Theorem \ref{thm:caching1}
holds only with probability $1-o(1)$ for different reasons in the fast
and slow fading case. For fast fading, the theorem holds only for node
placements that are ``regular'' enough. A random node placement
satisfies these regularity conditions with high probability as
$n\to\infty$. For slow fading, Theorem \ref{thm:caching1} holds under
the same regularity conditions on the node placement, but additionally
only holds with probability $1-o(1)$ for the realization of the channel
gains.

Given the equivalence of $\rhoca(\lambdaca), \phi(\lambdaca)$ and
$\Lambdaca(n), \gLambdaca(n)$ as mentioned above, Theorem
\ref{thm:caching1} states that $b_1(n) \gLambdaca(n) \subset
\Lambdaca(n)$ with high probability. This links the tree graph $G$ to
the wireless network: Every caching traffic matrix that can be routed
over the graph $G$ can also (up to a small, in the scaling sense,
factor) be transmitted reliably over the wireless network.

The communication scheme achieving the inner bound in Theorem
\ref{thm:caching1} consists of three layers. The lower two layers handle
channel coding and load balancing, and effectively transform the
wireless network into the tree graph $G$. The top layer assigns caches
to destination nodes and routes data over $G$. Thus, this scheme
performs separate source coding (in the top layer) and channel coding
(in the two bottom layers). See Section \ref{sec:main_scheme} for a
detailed description of this communication architecture.

\subsection{Outer Bound} 
\label{sec:main_outer}

The next result provides an outer bound for the caching capacity region
$\Lambdaca(n)$ in terms of the $\gLambdaca(n)$. This result is valid
for $\alpha >6$, i.e., for large path-loss exponents.

\begin{theorem}
    \label{thm:caching2}
    Under either fast or slow fading, for any $\alpha>6$, there exists
    $b_2(n) \leq n^{o(1)}$ such that 
    \begin{equation*}
        \rhoca(\lambdaca) \leq  b_2(n)\phi(\lambdaca)
    \end{equation*}
    for all $\lambdaca \in \Rp^{2^{n}\times n}$ 
    with probability $1-o(1)$ as $n\to\infty$.
\end{theorem}

The proof of Theorem \ref{thm:caching2} is provided in Section
\ref{sec:outer_proof}.  As with Theorem \ref{thm:caching1}, Theorem
\ref{thm:caching2} holds with probability $1-o(1)$ for the realization
of the node placement and, in the slow fading case, the realization of
the channel gains.

Using again the equivalence of $\rhoca(\lambdaca), \phi(\lambdaca)$ and
$\Lambdaca(n), \gLambdaca(n)$, Theorem \ref{thm:caching2} states that
$\Lambdaca(n) \subset b_2(n)\gLambdaca(n)$ with high probability.
Comparing Theorems \ref{thm:caching1} and \ref{thm:caching2}, we see
that, for $\alpha >6$ and with high probability, 
\begin{equation*}
    n^{-o(1)}\phi(\lambdaca)
    \leq \rhoca(\lambdaca) 
    \leq n^{o(1)}\phi(\lambdaca)
\end{equation*}
for all $\lambdaca \in \Rp^{2^{n}\times n}$ or, equivalently, 
\begin{equation*}
    n^{-o(1)}\gLambdaca(n)
    \subset \Lambdaca(n) 
    \subset n^{o(1)}\gLambdaca(n).
\end{equation*}
In other words, for $\alpha>6$, the set of
caching traffic matrices $\gLambdaca(n)$ supportable by routing over the
tree graph $G$ scales as the caching capacity region $\Lambdaca(n)$.
This is illustrated in Fig. \ref{fig:approx}.
\begin{figure}[tbp]
    \begin{center}
        \input{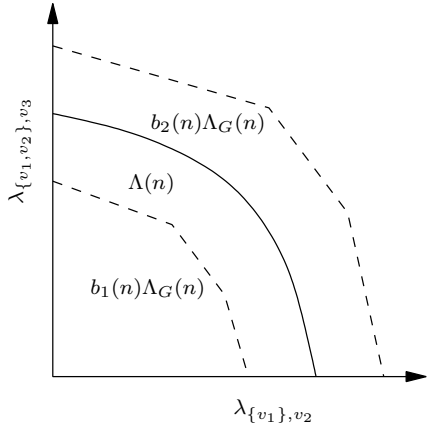}
    \end{center}

    \caption{ 
    For $\alpha >6$, the set $\gLambdaca(n)$ approximates the
    caching capacity region $\Lambdaca(n)$ of the wireless network in
    the sense that $b_1(n)\gLambdaca(n)$ (with $b_1(n)\geq n^{-o(1)}$)
    provides an inner bound to $\Lambdaca(n)$ and $b_2(n)\gLambdaca(n)$
    (with $b_2(n) \leq n^{o(1)}$) provides an outer bound to
    $\Lambdaca(n)$. The figure shows two dimensions (namely
    $\lambdaca_{\{v_1\},v_2}$ and $\lambdaca_{\{v_1,v_2\},v_3}$) of the $2^{n}\times
    n$-dimensional sets $\Lambdaca(n)$ and $\gLambdaca(n)$.  
    }

    \label{fig:approx}
\end{figure}

\subsection{Computational Aspects}
\label{sec:main_comput}

Theorems \ref{thm:caching1} and \ref{thm:caching2} show that, for large
$\alpha$, $\gLambdaca(n) \approx \Lambdaca(n)$. Computationally, the
question of interest is that of membership, i.e., determining if a given
$\lambdaca \in \Rp^{2^{n}\times n}$ belongs to $\Lambdaca(n)$ or,
equivalently, determining if $\rhoca(\lambdaca) \geq 1$. Since
$\rhoca(\lambdaca) \approx \phi(\lambdaca)$, computation of
$\phi(\lambdaca)$ answers the membership question approximately (up to a
multiplicative error of $n^{o(1)}$). 

The linear program \eqref{eq:computation} defining $\phi(\lambdaca)$ can
be solved in polynomial time in the number of its constraints and
variables \cite{kar}. Define 
\begin{equation*}
    \norm{\lambdaca}_0 
    \defeq \card{\{(U,w) : \lambdaca_{U,w} > 0 \}}
\end{equation*}
as the number of $(U,w)$ pairs with positive demand
$\lambdaca_{U,w} > 0$. The number of constraints in the linear program
\eqref{eq:computation} scales linearly in $\card{E_G} +
\norm{\lambdaca}_0$. And the number of variables scales as $n
\norm{\lambdaca}_0$. Noting that $\card{E_G}$ is polynomial in $n$ by
\eqref{eq:cardvg}, this implies that the approximate membership of any
$\lambdaca$ in $\Lambdaca(n)$ can be checked in time polynomial in $n$
and $\norm{\lambdaca}_0$.  

Note that this need not be polynomial in $n$, since $\norm{\lambdaca}_0$
could be exponential in $n$. However, even just to ask the membership
query, one needs to specify $\norm{\lambdaca}_0$ distinct numbers.
Therefore, the above discussion shows that the computational cost of
approximate membership testing takes time polynomial in the effective
problem statement, which is the best one can hope for. Moreover, in many
situations of practical interest, the number of pairs $(U,w)$ with
positive demand can be expected to be only polynomial in the network
size $n$. In these cases, approximate membership can be tested in
polynomial time also in $n$.

\subsection{A Content Delivery Protocol}
\label{sec:main_scheme}

Theorem \ref{thm:caching1} provides an inner bound for the caching
capacity region of a wireless network. We now describe the
communication scheme achieving this inner bound. The matching outer bound
shows that, for $\alpha>6$, this scheme is optimal in the scaling sense.

Our proposed communication scheme consists of three layers, similar to a
protocol stack. From the highest to lowest level of abstraction, these
three layers are the \emph{data layer}, the \emph{cooperation
layer}, and the \emph{physical layer}. 

From the view of the data layer, the wireless network is
treated as the abstract capacitated tree graph $G$, up to a loss of a
factor $b_1(n)$ in the capacity of each link.  Let us assume that
$\frac{1}{b_1(n)}\lambdaca \in \gLambdaca(n)$.  Solve the corresponding
linear program \eqref{eq:computation}, and let $f=(f_{p,U})$ be its
solution.  Since $\frac{1}{b_1(n)}\lambdaca \in \gLambdaca(n)$, routing
traffic according to this solution $f$ allows to support the caching
traffic matrix $\lambdaca$ in this layer.  The next two layers transform
this routing solution $f$ for $\lambdaca$ over the graph $G$ into a
communication strategy for the wireless network. 

The cooperation layer provides this tree graph abstraction to the data
layer. Recall that the leaf nodes of $G$ are the nodes $V(n)$ of the
wireless network and that each internal node of $G$ represents a subset
of nodes $V_{\ell,i}(n)\subset V(n)$ within the subsquare
$A_{\ell,i}(n)$ in the wireless network.  The cooperation layer provides
the tree abstraction $G$ by ensuring that, whenever a message is located
in the data layer at a particular node $v$, the message is evenly 
distributed in the wireless network among the nodes $V_{\ell,i}(n)$
represented by the node $v$. Recall that the sets $\{V_{\ell,i}(n)\}$
are nested and increasing as $\ell$ decreases.  Hence, as a message
travels towards the root node in $G$ in the data layer, it is
distributed over a larger area in the wireless network by the
cooperation layer.  Similarly, as a message travels away from the root
node in $G$ in the data layer, it is concentrated on a smaller area in
the wireless network by the cooperation layer.  Thus, sending a message
up or down an edge in the tree $G$ in the data layer corresponds in the
cooperation layer to distributing or concentrating the same message in
the wireless network (see also Fig.~\ref{fig:layers} below). 

Formally, this distribution and concentrating of messages is performed
as follows. To send a message from a child node to its parent in $G$
(i.e., towards the root node of $G$), the message at the wireless nodes
in $V(n)$ represented by the child node in $G$ is evenly distributed
over the wireless channel among all nodes in $V(n)$ represented by the
parent node in $G$. This distribution is performed by splitting the
message at each node in $V(n)$ represented by the child node in $G$ into
equal sized parts and by transmitting one part to each node in $V(n)$
represented by the parent node in $G$. To send a message from a parent
node to a child node in $G$ (i.e., away from the root node of $G$), the
message at the wireless nodes in $V(n)$ represented by the parent node
in $G$ is concentrated on the wireless nodes in $V(n)$ represented by
the child node in $G$. This concentration is performed be collecting at
each node in $V(n)$ corresponding to the child node in $G$ the message
parts of the previously split up message located at the nodes in $V(n)$
corresponding to the parent node in $G$.

Finally, the physical layer performs this concentration and distribution
of messages induced by the cooperation layer over the physical wireless
channel.  Note that the kind of traffic resulting from the operation of
the distribution or cooperation is highly uniform in the sense that within each
subsquare all nodes receive data at the same rate. Uniform traffic of
this sort is well understood. Depending on the path-loss exponent
$\alpha$, we use either hierarchical cooperation \cite{nie,ozg} (for
$\alpha\in(2,3]$) or multi-hop communication (for $\alpha>3$). It is
this operation of each edge in the physical layer that determines the
edge capacity of the graph $G$ as seen from the data layer. 

Note that the value of the path-loss exponent $\alpha$ only
significantly affects the operation of the physical layer. The
cooperation layer is completely invariant under changes in $\alpha$, and
the data layer is only affected through the value of the edge
capacities of the graph $G$.  In particular, even when $\alpha>3$ so
that the physical layer performs multi-hop communication, the
construction of the tree structure $G$ is still necessary. In fact, the
role of routing over $G$ can be understood as load balancing of traffic,
which is required no matter how the physical layer operates.

We point out that this scheme respects source-channel separation. In
fact, source coding is only performed at the data layer (through the
selection of message parts from the various available caches). Channel
coding is only performed in the cooperation and physical layers.

The next example illustrates the operation of this scheme.  For more
details on this architecture, see \cite{nie2}.
\begin{example}
    Consider the three layers of the proposed communication architecture
    depicted in Fig. \ref{fig:layers}. From top to bottom in the figure,
    these are the data layer, the cooperation layer, and the physical
    layer. In this example, we consider a single $(U,w)$ pair.  The set
    of caches $U$ consists of two nodes $\{u_1, u_2\}$ in the wireless
    network shown at the bottom left, and the corresponding destination
    $w$ is in the top right of the network. 
    
    At the data layer, traffic is balanced by choosing which fraction of
    the message requested at $w$ and available at $U$ is delivered from
    each node $u_1$ and $u_2$ in $U$. This load balancing is performed
    by solving the linear program \eqref{eq:computation}. In this simple
    example, a reasonable choice is to deliver half the message from
    $u_1$ and half from $u_2$. The routes between $\{u_1, u_2\}$ and $w$
    chosen at the data layer are indicated in black dashed lines. 
    
    Consider now the second edge along the path in $G$ from $u_1$ to $w$
    labeled by $e=(v_2,v_1)$ in the figure.  The middle plane in the
    figure shows the induced behavior in the cooperation layer from
    using this edge in the data layer. Note that $v_2$ and $v_1$ are not
    leaf nodes of $G$, and hence correspond to subsets of $V(n)$ through
    the construction of $G$. Let $V_{2,i}(n)$ and $V_{1,j}(n)$ be the
    subsets of $V(n)$ corresponding to $v_2$, and $v_1$, respectively.
    Since $v_2$ is a child node of $v_1$, we must have $V_{2,j}\subset
    V_{1,i}$. When a message is present at $v_2$ in the data layer, it
    is distributed evenly over the three nodes in $V_{2,i}(n)$ in the
    cooperation layer; in other words, each of the three nodes in
    $V_{2,i}(n)$ has access to a distinct third of the original message.
    To send the message over edge $e$ from $v_2$ to $v_1$ in the data
    layer, the cooperation layer splits the message part at each node in
    $V_{2,i}(n)$ into smaller parts and distributes these subparts
    evenly over the nodes in $V_{1,j}(n)$. Thus, when the message
    reaches $v_1$ in the data layer, each of the nine nodes in
    $V_{1,j}(n)$ has access to a distinct ninth of the original message
    in the cooperation layer. 

    The bottom plane in the figure shows part of the corresponding
    actions induced in the physical layer. The distribution of message
    parts from $V_{2,i}(n)$ to $V_{1,j}(n)$ is properly scheduled to
    minimize interference, and channel coding is performed. The precise
    nature of the operation of this layer depends on the path-loss
    exponent $\alpha$, as explained above.

    \begin{figure}[!ht]
        \begin{center}
            \scalebox{0.89}{
            \input{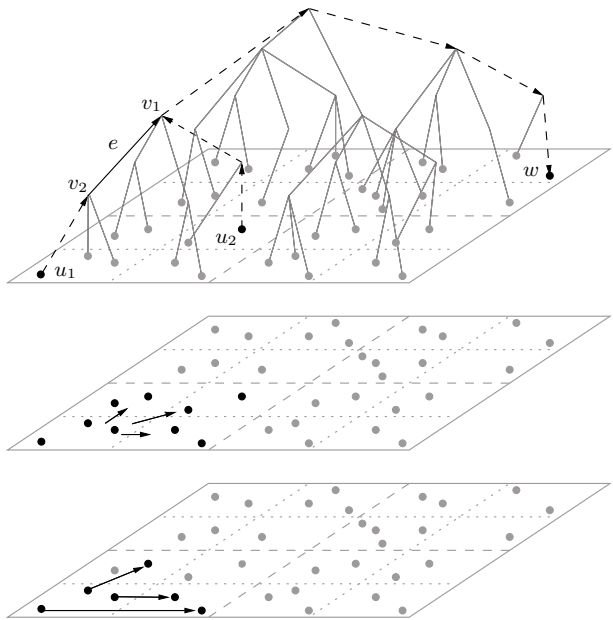}
            }
        \end{center}

        \caption{Example operation of the three-layer architecture. 
        A message available at the caches $U=\{u_1,u_2\}$ is requested at
        the destination node $w$. The figure shows the induced actions
        by this request in the data layer (top plane), cooperation layer
        (middle plane), and physical layer (bottom plane).}

        \label{fig:layers}
    \end{figure}
\end{example}

\subsection{Example Scenarios}
\label{sec:main_examples}

We provide two examples illustrating various aspects of the
caching capacity region. 
Example \ref{eg:nearest} shows that the strategy of always selecting the
nearest cache can be arbitrarily bad. Example \ref{eg:complete}
illustrates the potential benefit of caching on achievable rates 
in the wireless network.

\begin{example}
    \label{eg:nearest}
    ({\em Nearest-neighbor cache selection})

    A simple and intuitive strategy for selecting caches is to request
    the entire message from the nearest available cache. In fact, this
    is the strategy implicitly assumed in most of the prior work
    on caching in wireless networks cited in Section
    \ref{sec:prior}. This example shows that this strategy can be
    arbitrarily bad.

    We consider the scenario illustrated in Fig. \ref{fig:nearest}. 
    Assume $V_{2,1}(n)$ and $V_{2,3}(n)$ are subsets of $V_{1,1}(n)$,
    and $V_{2,16}(n)$ is a subset of $V_{1,4}(n)$.  Consider a node
    $u^\star\in V_{2,3}(n)$ geographically close to $V_{2,1}(n)$, and label
    the nodes in $V_{2,1}(n) = \{w_1,w_2,\ldots\}$ and in $V_{2,16}(n) =
    \{u_1,u_2,\ldots\}$.  Construct the traffic matrix
    \begin{equation*}
        \lambdaca_{U,w} \defeq
        \begin{cases}
            1 & \text{if $U=\{u^\star,u_i\}, w = w_i$ for some $i$,} \\
            0 & \text{otherwise.}
        \end{cases}
    \end{equation*}
    In words, each node in $w_i\in V_{2,1}(n)$
    requests a message available at a dedicated cache $u_i\in V_{2,16}$
    and at a shared cache $u^\star\in V_{2,3}$. We want to determine $\rho(\lambdaca)$, the
    largest multiple of $\lambdaca$ such that the resulting traffic matrix is
    achievable in the wireless network. In this setting with unit
    demands, $\rho(\lambdaca)$ can also be interpreted as the largest
    uniformly achievable per-node rate.
    \begin{figure}[!ht]
        \begin{center}
            \scalebox{0.8}{
            \input{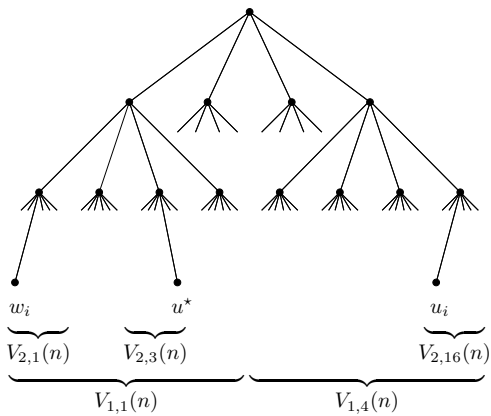}
            }
        \end{center}

        \caption{Caching traffic pattern for Example \ref{eg:nearest}.
        Each destination node $w_i\in V_{2,1}(n)$ requests a message 
        available at a dedicated cache $u_i\in V_{2,16}(n)$ and at a 
        shared cache $u^\star\in V_{2,3}(n)$.
        }

        \label{fig:nearest}
    \end{figure}

    For every destination node $w_i$, the nearest cache (both in terms
    of geographic as well as graph distance) is $u^\star$. Assume each
    node $w_i$ requests the entire message from its nearest cache
    $u^\star$. It is easy to show that each node in the wireless
    network, and, in particular, node $u^\star$, can reliably transmit
    information at a sum rate of at most $n^{o(1)}$.  With high
    probability there will be $\Theta(n)$ nodes in $V_{2,1}(n)$
    requesting a message at equal rate from $u^\star$.  Hence this
    strategy achieves a per-node rate of at most $n^{-1+o(1)}$
    regardless of the value of the path-loss exponent $\alpha >2$. 

    Assume now each $w_i$ uses only the more distant cache $u_i$. The
    routes from $u_i$ to $w_i$ for different values of $i$ intersect
    only at the four edges closest to the root node of $G$. These four
    edges have a capacity of order $\Theta(n^{2-\min\{3,\alpha\}/2})$,
    and hence it can be seen that over the graph $G$ these messages can
    be routed at a per-node rate of $\Theta(n^{1-\min\{3,\alpha\}/2})$. 
    Together with Theorem \ref{thm:caching1}, this shows that
    \begin{equation*}
        \rho(\lambdaca) 
        \geq n^{1-\min\{3,\alpha\}/2-o(1)}
        \gg n^{-1+o(1)}
    \end{equation*}
    is achievable in the wireless network with high probability. For
    this simple example, it is easily checked that this strategy is
    order-optimal for routing over the graph $G$. Together with Theorem
    \ref{thm:caching2}, this confirms that, for $\alpha >6$, no scheme
    can achieve a better scaling in the wireless network. Hence
    \begin{equation*}
        \rho(\lambdaca) = n^{1-\min\{3,\alpha\}/2\pm o(1)}
    \end{equation*}
    for $\alpha > 6$.\footnote{The notation $n^{\pm o(1)}$ is used to
    indicate that $n^{o(1)}$ is an upper and $n^{-o(1)}$ is a lower bound.} 
    With some additional work, it can be shown that
    this is the correct scaling of $\rho(\lambdaca)$ also for $\alpha\in
    (2,6]$. This shows that the strategy of always selecting the nearest
    cache can result in a scaling exponent that is considerably worse
    than what is achievable with optimal cache selection.
\end{example}

\begin{example}
    \label{eg:complete}
    ({\em Complete caches})

    Assume we randomly pick  $n^\beta$ caches for $\beta\in(0,1)$, each
    holding a complete copy of all the messages. More precisely, letting
    $U^\star = \{u_i\}_{i=1}^{n^\beta}$ be the collection of
    caches, we consider a caching traffic matrix
    $\lambdaca\in\Rp^{2^n\times n}$ of the form 
    \begin{equation*}
        \lambdaca_{U,w} =
        \begin{cases}
            1 & \text{if $U=U^\star$,} \\
            0 & \text{otherwise,}
        \end{cases}
    \end{equation*}
    for every $(U,w)$. In other words, every node $w\in V(n)$ requests
    a message that is available at a common set of caches $U^\star$.
    As before, $\rho(\lambdaca)$ can in this setting with uniform
    demands be interpreted as the largest uniformly achievable per-node
    rate.
    
    Assume every node chooses the nearest cache (as discussed in Example
    \ref{eg:nearest}). With high probability, there will be
    $\Theta(n^{1-\beta})$ nodes accessing the same cache. The bottleneck
    limiting the flows from this cache to the destination nodes is the
    edge with capacity one connecting the cache to the tree.
    Hence, with this strategy, we can achieve a per-node rate of
    $\Theta(n^{\beta-1})$ over the graph $G$ with high probability.
    By Theorem \ref{thm:caching1}, this implies that
    a per-node rate of 
    \begin{equation*}
        \rho(\lambdaca) \geq n^{\beta-1-o(1)}
    \end{equation*}
    is achievable with probability $1-o(1)$ as $n\to\infty$ in the
    wireless network.  A short calculation reveals that this is an
    order-optimal routing strategy over $G$, which, by Theorem
    \ref{thm:caching2}, shows that 
    \begin{equation*}
        \rho(\lambdaca) \leq n^{\beta-1+o(1)}
    \end{equation*}
    for $\alpha > 6$.  Hence, for $\alpha >6$, 
    \begin{equation*}
        \rho(\lambdaca) = n^{\beta-1\pm o(1)}.
    \end{equation*}
    Moreover, it can be shown that this is the correct scaling of
    $\rho(\lambdaca)$ also for $\alpha\in (2,6]$.

    This example illustrates that in situations in which the traffic
    demand and location of caches are regular enough, the strategy of
    selecting the nearest cache (as analyzed also in Example
    \ref{eg:nearest}, and which is shown there to be arbitrarily bad in
    general) can actually be close to optimal.
\end{example}

\section{Proofs}
\label{sec:proofs}

In Section \ref{sec:inner_proof}, we provide the proof of the inner
bound in Theorem \ref{thm:caching1}. The proof relies on the
communication scheme presented earlier in Section \ref{sec:main_scheme}.
The outer bound in Theorem \ref{thm:caching2} is proved in Section
\ref{sec:outer_proof}. It consists of two key steps, summarized by
Lemmas \ref{thm:one} and \ref{thm:two} below. The first step is
information-theoretic, outer bounding the caching capacity region in
terms of cuts in the wireless network and then relating these cuts to
cuts in the graph $G$. The details of this first step are provided in
Section \ref{sec:outer_one}. The second step relates these cuts in the
graph $G$ to supportable flows over $G$. The details of this second step
are provided in Section \ref{sec:outer_two}.

\subsection{Proof of Theorem \ref{thm:caching1} (Inner Bound)}
\label{sec:inner_proof}

We wish to show that $\rho(\lambdaca) \geq b_1(n)\phi(\lambdaca)$ for
some $b_1(n) \geq n^{-o(1)}$ uniform in $\lambdaca$.  Equivalently, we
will argue that $\lambdaca \in \gLambdaca(n)$ implies $b_1(n) \lambdaca
\in \Lambdaca(n)$. Assume $\lambdaca \in \gLambdaca(n)$; then
$\phi(\lambdaca)\geq 1$. Let 
\begin{equation*}
    f\defeq (f_{p,U})
\end{equation*}
be the corresponding solution of the linear program
\eqref{eq:computation}. By definition of \eqref{eq:computation}, the
load induced by $f$ on each edge of $G$ is no more than its
capacity. 

We now use this solution $f$ to construct a \emph{unicast} traffic
matrix solving the caching problem. Formally, a \emph{unicast traffic
matrix} is an element $\lambdauc\in\Rp^{n\times n}$ associating with
each source-destination pair $(u,w)\in V(n)\times V(n)$ the rate
$\lambdauc_{u,w}$ at which destination node $w$ requests data from
source node $u$. The \emph{unicast capacity region}
$\Lambda(n)\subset\Rp^{n\times n}$ is the closure of the collection of
all achievable unicast traffic matrices in the wireless network. 
In analogy to caching traffic, every unicast traffic matrix $\lambdauc$
for the wireless network induces a unicast traffic matrix between the
leaf nodes of the graph $G$, and we can define $\gLambdauc(n)$ as the
collection of unicast traffic matrices that can be routed (i.e., are
supportable) over $G$. 

Consider again the flows $f$ as defined above. Construct the
unicast traffic matrix $\lambdauc = \lambdauc(f)$ as
\begin{equation*}
    \lambdauc_{u,w} \defeq \sum_{\substack{U\subset V(n): \\ u\in U}} f_{p_{u,w},U},
\end{equation*}
where $p_{u,w}$ is the unique path in the tree graph $G$ between $u$ and $w$.
In words, $\lambdauc_{u,w}$ is the sum of the flows $f_{p_{u,w},U}$ for the
caching problem from $u$ to $w$. The load induced by this unicast
traffic $\lambdauc(f)$ on the edges of $G$ is the same as that due to
$f$. In particular, the total demand of $\lambdauc(f)$ across each
edge is at most its capacity. Since $G$ is a tree, this implies that
$\lambdauc(f)$ is supportable over $G$, i.e.,
$\lambdauc(f)\in\gLambdauc(n)$.

We have thus transformed the problem of routing \emph{caching} traffic
over $G$ into one of routing \emph{unicast} traffic over $G$. The
following result, established in \cite{nie2}, links the set of
supportable unicast traffic matrices $\gLambdauc(n)$ over $G$ to the
unicast capacity region $\Lambdauc(n)$ of the wireless network.
\begin{proposition}
    \label{prop:ngs}
    Under either fast or slow fading, for any $\alpha>2$,
    there exists $b_1'(n) \geq n^{-o(1)}$ such that
    \begin{equation*}
        b_1'(n) \gLambdauc(n) \subset \Lambdauc(n)
    \end{equation*}
    with probability $1-o(1)$ as $n\to\infty$.
\end{proposition}
\begin{IEEEproof}
    See \cite[Lemma 10]{nie2}.
\end{IEEEproof}
Proposition \ref{prop:ngs} is established by means of an explicit
communication architecture, consisting of the three layers (data layer,
cooperation layer, physical layer) as described in detail in Section
\ref{sec:main_scheme}.

Proposition \ref{prop:ngs} implies that $b_1'(n) \lambdauc(f) \in
\Lambdauc(n)$. Given that the unicast traffic matrix $\lambdauc(f)$
was created through decomposing the caching traffic matrix $\lambdaca$,
it follows that $b_1'(n)\lambdaca$ can be supported using these unicast
transmissions over the wireless network. That is, $b_1(n)\lambdaca \in
\Lambdaca(n)$ for 
\begin{equation*}
    b_1(n) \defeq b_1'(n) \geq n^{-o(1)}.
\end{equation*}
This shows that
\begin{equation*}
    b_1(n)\gLambdaca(n) \subset \Lambdaca(n),
\end{equation*}
completing the proof of Theorem \ref{thm:caching1}.\hfill\IEEEQED

\subsection{Proof of Theorem \ref{thm:caching2} (Outer Bound)}
\label{sec:outer_proof}

We aim to show that 
\begin{equation*}
    \rho(\lambda)\leq b_2(n)\phi(\lambda)
\end{equation*}
for some $b_2(n) \leq n^{o(1)}$ uniform in $\lambda$. The proof proceeds
in two steps. First, we relate achievable traffic in the wireless
network (characterized by $\rho(\lambda)$) to cuts in the graph $G$
(characterized by $\hat{\rho}(\lambda)$ defined below). Second, we
relate these cuts in $G$ to supportable flows over $G$ (characterized by
$\phi(\lambda)$).

Define 
\begin{align*}
    \hLambdaca(n) \defeq 
    \bigg\{ \lambdaca\in\Rp^{2^n\times n}: 
    \sum_{U\subset S\cap V}\sum_{w\in V\setminus S} \lambdaca_{U,w} 
    \leq \sum_{\substack{(u,v)\in E_G: \\ u\in S, v\notin S}} c_{u,v} 
    \ \forall S\subset V_G \bigg\}, 
\end{align*}
with $V=V(n)$ and $V_G=V_G(n)$. Furthermore, let, for any caching
traffic matrix $\lambdaca\in\Rp^{2^n\times n}$,
\begin{equation}
    \label{eq:rhohatdef}
    \hat{\rho}(\lambdaca) 
    \defeq \max\big\{\rho\geq 0: \rho\lambdaca \in \hLambdaca(n)\big\}.
\end{equation}
The set $\hLambdaca(n)$ corresponds to the restrictions on 
the set of supportable caching traffic matrices on the graph $G$ by
all possible cuts $S$ in $V_G(n)$. Consider one such cut $S\subset
V_G(n)$. For any caching traffic matrix $\lambdaca$ that can be routed
over $G$, the total flow
\begin{equation*}
    \sum_{U\subset S\cap V(n)}\sum_{w\in V(n)\setminus S} \lambdaca_{U,w}
\end{equation*}
across this cut can not be larger than the capacity of the cut
\begin{equation*}
    \sum_{\substack{(u,v)\in E_G: \\ u\in S, v\notin S}} c_{u,v}.
\end{equation*}
The region $\hLambdaca(n)$ is the set of caching traffic matrices
satisfying all these constraints. The scalar $\hat{\rho}(\lambdaca)$
yields an equivalent description of $\hLambdaca(n)$. Note that we can rewrite
the definition of $\hat{\rho}(\lambdaca)$ as
\begin{equation}
    \label{eq:rhohat}
    \hat{\rho}(\lambdaca) 
    = \min_{S\subset V_G(n)} 
    \frac{\sum_{\substack{(u,v)\in E_G: \\ u\in S,v\notin S}} c_{u,v}}
    {\sum_{U\subset S\cap V(n)}\sum_{w\in V(n)\setminus S} \lambdaca_{U,w}}.
\end{equation}

Recall that $\gLambdaca(n)$ is the set of supportable caching traffic
matrices on $G$, and that $\phi(\lambdaca)$ is its equivalent description.
From the discussion in the last paragraph, it is clear that
$\gLambdaca(n)\subset\hLambdaca(n)$, or, equivalently, that
$\phi(\lambdaca) \leq \hat{\rho}(\lambdaca)$. The next lemma shows that
$\hat{\rho}(\lambdaca)$ is also an approximate upper bound on the
equivalent description $\rho(\lambdaca)$ of the caching capacity region
$\Lambdaca(n)$ of the wireless network.  

\begin{lemma}
    \label{thm:one}
    Under either fast or slow fading, for any $\alpha > 6$, there exists
    $b_3(n) \leq n^{o(1)}$ such that
    \begin{equation*}
        \rho(\lambdaca)
        \leq b_3(n) \hat{\rho}(\lambdaca)
    \end{equation*}
    for all caching traffic matrices $\lambdaca\in\Rp^{2^n\times n}$
    with probability $1-o(1)$ as $n\to\infty$.
\end{lemma}
The proof of Lemma \ref{thm:one} is presented in Section
\ref{sec:outer_one}.

Lemma \ref{thm:one} shows that, for $\alpha > 6$, $\Lambdaca(n)\subset
b_3(n)\hLambdaca(n)$. This implication is much less obvious than the
statement $\gLambdaca(n)\subset\hLambdaca(n)$. The proof of Lemma
\ref{thm:one} first uses the information-theoretic cut-set bound to
upper bound achievable rates for caching traffic by cuts in the wireless
network and then relates these cuts in the wireless network to cuts in
the graph $G$. We point out that it is this step that limits the
applicability of the outer bound in Theorem \ref{thm:caching2} to large
path-loss exponents $\alpha > 6$. The reason for this is that evaluation
of the cut-set bound for the wireless network for small path-loss
exponents is quite difficult. While it is known how to evaluate
``rectangular'' cuts for small $\alpha$ \cite{ozg}, these techniques do
not extend to the arbitrary cuts that are required for the analysis of
caching traffic. 

Lemma \ref{thm:one} allows us to upper bound the equivalent description
$\rho(\lambdaca)$ of the caching capacity region $\Lambdaca(n)$ by the
equivalent description $\hat{\rho}(\lambdaca)$ of the set $\hLambdaca(n)$ of
caching traffic matrices satisfying all cut constraints in the graph
$G$. We now show that $\hat{\rho}(\lambdaca)$ can be upper bounded by
the equivalent description $\phi(\lambdaca)$ of the set $\gLambdaca(n)$ of
supportable caching traffic matrices on $G$. 

\begin{lemma}\label{thm:two}
    For any $\alpha>2$, there exists $b_4(n) \geq n^{-o(1)}$ such that
    \begin{equation*}
        b_4(n) \hat{\rho}({\lambdaca}) \leq \phi(\lambdaca)
    \end{equation*}
    for all caching traffic matrices $\lambdaca\in\Rp^{2^n\times n}$.
\end{lemma}
The proof of Lemma \ref{thm:two} is presented in Section
\ref{sec:outer_two}.

Lemma \ref{thm:two} shows that, for any $\alpha > 2$,
$b_4(n)\hLambdaca(n)\subset \gLambdaca(n)$. From the above discussion,
we already know that $\gLambdaca(n)\subset \hLambdaca(n)$. Hence, we
deduce from Lemma \ref{thm:two} that
$\gLambdaca(n)\approx\hLambdaca(n)$. This can be understood as an
approximate max-flow min-cut result for caching traffic on undirected
capacitated graphs. Lemma \ref{thm:two} is, in fact, valid for any tree
graph $G$ (with mild assumptions on the edge capacities, see the proof
for the details) and might be of independent interest. 

Combining Lemmas \ref{thm:one} and \ref{thm:two} shows that, for any
$\alpha > 6$, 
\begin{align*}
    \rho(\lambdaca) 
    & \leq b_3(n)\hat{\rho}(\lambda) \\
    & \leq \frac{b_3(n)}{b_4(n)} \phi(\lambdaca).
\end{align*}
Setting
\begin{equation*}
    b_2(n) \defeq b_3(n)/b_4(n) \leq n^{o(1)},
\end{equation*}
and noting that $b_2(n)$ is uniform in $\lambdaca$, concludes the proof of Theorem
\ref{thm:caching2}.

\subsection{Proof of Lemma \ref{thm:one}}
\label{sec:outer_one}

We start with several auxiliary results.  We first introduce some
regularity conditions that are satisfied with high probability by a
random node placement. Define $\mc{V}(n)$ to be the collection of all
node placements $V(n)$ that satisfy the following conditions:
\begin{equation*}
    \begin{array}{r@{\,}ll}
        r_{u,v} & > n^{-1} 
        & \text{ for all $u,v\in V(n), u\neq v$,} \\
        \card{V_{\ell,i}(n)} & \leq \log(n)  
        & \text{ for $\ell= \frac{1}{2}\log(n)$ and all $i\in\{1,\ldots, 4^\ell\}$,} \\
        \card{V_{\ell,i}(n)} & \geq 1 
        & \text{ for $\ell=\frac{1}{2}\log\big(\frac{n}{2\log(n)}\big)$ 
        and all $i\in\{1,\ldots, 4^\ell\}$,} \\
        \card{V_{\ell,i}(n)} & \in [4^{-\ell-1}n,4^{-\ell+1}n]
        & \text{ for all }
        \ell \in\big\{1,\ldots,\frac{1}{2}\log(n)\big(1-\log^{-5/6}(n)\big)\big\},
        i\in\{1,\ldots,4^\ell\}.
    \end{array}
\end{equation*}
The first condition is that the minimum distance between node pairs is
not too small. The second condition is that all squares of area $1$
contain at most $\log(n)$ nodes. The third condition is that all squares
of area $2\log(n)$ contain at least one node. The fourth condition is
that all squares up to level
$\frac{1}{2}\log(n)\big(1-\log^{-5/6}(n)\big)$ contain a number of nodes
proportional to their area. 

The next lemma, quoted from \cite{nie2}, states that a random node
placement satisfies these conditions with high probability.
\begin{lemma}
    \label{thm:mcv}
    \begin{equation*}
        \Pp\big(V(n)\in\mc{V}(n)\big) \geq 1-o(1)
    \end{equation*}
    as $n\to\infty$.
\end{lemma}
\begin{IEEEproof}
    See \cite[Lemma 5]{nie2}.
\end{IEEEproof}

We continue with results upper bounding the MIMO capacity between
subsets of nodes in $V(n)$. Formally, for disjoint subsets $S_1,S_2\subset
V(n)$, denote by $C(S_1,S_2)$ the MIMO capacity between the nodes in
$S_1$ and $S_2$.
Let
\begin{equation*}
    \bm{H}_{S_1,S_2} \defeq (h_{u,v})_{u\in S_1,v\in S_2}
\end{equation*}
be the matrix of channel gains between the nodes in $S_1$ and $S_2$.
Under fast fading,
\begin{equation*}
    C(S_1,S_2) \defeq 
    \max
    \E\Big( \log \det\big(\bm{I}+\bm{H}_{S_1,S_2}^\dagger\bm{Q}(\bm{H}) 
    \bm{H}_{S_1,S_2}\big)\Big),
\end{equation*}
where the maximization is over all positive semi-definite matrices
$\bm{Q}(\bm{H})$ such that $\E(q_{u,u})\leq 1$ for all $u\in
S_1$. Under slow fading,
\begin{equation*}
    C(S_1,S_2) 
    \defeq \max
    \log \det\big(\bm{I}+\bm{H}_{S_1,S_2}^\dagger\bm{Q} \bm{H}_{S_1,S_2}\big),
\end{equation*}
where the maximization is over all positive semi-definite matrices
$\bm{Q}$ such that $q_{u,u}\leq 1$ for all $u\in S_1$.  See, e.g.,
\cite{tse}. To simplify notation, define furthermore
\begin{equation*}
    r_{S,v} \defeq \min_{u\in S}r_{u,v}
\end{equation*}
for $S\subset V(n)$. The next lemma provides an upper bound on the MIMO
capacity $C(S,S^c)$ between the nodes in $S$ and $S^c$ in terms of the
number of nodes close to the boundary between them.

\begin{lemma}
    \label{thm:mimo}
    Under either fast or slow fading, for every $\alpha > 6$, there
    exists a constant $K_1$ such that for large enough $n$ and all
    $V(n)\in\mc{V}(n)$ and $S\subset V(n)$
    \begin{equation*}
        C(S,S^c) 
        \leq K_1 \log^4(n)\big|\{v\in S^c: r_{S,v}< \log(n)+1\}\big|.
    \end{equation*}
\end{lemma}
\begin{IEEEproof}
    Set $S_1 \defeq S$ and $S_2 \defeq S^c$, and denote by $S_2^k$ the
    nodes in $S_2$ that are at distance between $k$ and $k+1$ from
    $S_1$, i.e., 
    \begin{equation*}
        S_2^k
        \defeq \Bigl\{v\in S_2: r_{S_1,v} \in [k,k+1)\Bigr\}.
    \end{equation*}
    Note that
    \begin{equation*}
        S_2 = \bigcup_{k=0}^{\infty} S_2^k
    \end{equation*}
    and
    \begin{equation*}
        \big|\{v\in S_2: r_{S_1,v}< \log(n)+1\}\big|
        = \sum_{k=0}^{\log(n)}\card{S_2^k}.
    \end{equation*}

    Applying the generalized Hadamard inequality, we obtain that
    under either fast or slow fading
    \begin{equation}
        \label{eq:mimo2}
        C(S_1,S_2) 
        \leq C\Bigl(S_1,\cup_{k=0}^{\log(n)}S_2^k\Bigr)
        +C\Bigl(S_1, \cup_{k>\log(n)} S_2^k\Bigr).
    \end{equation}

    For the first term in \eqref{eq:mimo2}, using Hadamard's
    inequality once more yields
    \begin{align*}
        C\Bigl(S_1,\cup_{k=0}^{\log(n)}S_2^k\Bigr)
        & \leq \sum_{k=0}^{\log(n)}\sum_{v\in S_2^k}C(S_1,\{v\}) \\
        & \leq \sum_{k=0}^{\log(n)}\sum_{v\in S_2^k}C(\{v\}^c,\{v\}).
    \end{align*}
    By \cite[Lemma~6]{nie2}, 
    \begin{equation*}
        C(\{v\}^c,\{v\}) \leq K \log(n)
    \end{equation*}
    for some constant $K < \infty$ depending only on $\alpha$, 
    and thus
    \begin{equation}
        \label{eq:mimo3}
        C\Bigl(S_1,\cup_{k=0}^{\log(n)}S_2^k\Bigr)
        \leq K\log(n)\sum_{k=0}^{\log(n)}\card{S_2^k}.
    \end{equation}

    For the second term in \eqref{eq:mimo2}, we have the following upper
    bound from slightly adapting \cite[Theorem~2.1]{jov}: Under
    either fast or slow fading,
    \begin{equation*}
        C\Big(S_1, \cup_{k>\log(n)} S_2^k\Big)
        \leq \sum_{k>\log(n)} \sum_{v\in S_2^k} 
        \Big(\sum_{u \in S_1}r_{u,v}^{-\alpha/2}\Big)^2.
    \end{equation*}
    By definition of $S_2^k$, for $v\in S_2^k$, the (open) disk of
    radius $k$ around $v$ does not contain any node in $S_1$. Moreover,
    since $V\in\mc{V}$, there are at most $\log(n)$ nodes inside every
    subsquare of $A$ of sidelength one.\footnote{To simplify notation,
    we suppress dependence of $V(n), \mc{V}(n), \Lambdaca(n), \ldots$
    within proofs whenever this dependence is clear from the context.}
    Thus, given that $\alpha > 6$, we have for any $v\in S_2^k$,
    \begin{align*}
        \sum_{u \in S_1}r_{u,v}^{-\alpha/2}
        & \leq \log(n)\sum_{\tilde{k} = k}^{\infty}10\pi (\tilde{k}+3) 
        \tilde{k}^{-\alpha/2} \\
        & \leq \tilde{K} \log(n) k^{2-\alpha/2},
    \end{align*}
    for some constant $\tilde{K} < \infty$ depending only on $\alpha$. 
    Therefore,
    \begin{equation}
        \label{eq:mimo4}
        C\Big(S_1, \cup_{k>\log(n)} S_2^k\Big)
        \leq \sum_{k>\log(n)} \card{S_2^k} \tilde{K}^2 \log^2(n) k^{4-\alpha}.
    \end{equation}

    Consider now some $v\in S_2^k$ with $k> \log(n)$, and let $u$ be the
    closest node in $S_1$ to $v$. Since $v\in S_2^k$, we must have
    \begin{equation*}
        r_{u,v}\in [k,k+1).
    \end{equation*}
    Consider the (open) disk of radius $r_{u,v}$ around $v$ and the disk
    of radius $\log(n)$ around $u$. Since $u$ is the closest node to $v$
    in $S_1$, all nodes in the disk around $v$ are in $S_2$.  Moreover,
    the intersection of the two disks has an area of at least
    $\frac{\pi}{4}\log^2(n)$. Since $V \in \mc{V}$, this implies that,
    for $n$ large enough, this intersection must contain at least one
    point, say $\tilde{v}$, and by construction
    \begin{equation*}
        \tilde{v}\in \bigcup_{\tilde{k}=0}^{\log(n)}S_2^{\tilde{k}}.
    \end{equation*}
    This shows that for every node $v$ in $S_2^k$ there exists a node
    $\tilde{v}$ in $\cup_{\tilde{k}=0}^{\log(n)}S_2^{\tilde{k}}$ such
    that
    \begin{equation*}
        r_{v,\tilde{v}} \in [k-\log(n),k+1).
    \end{equation*}
    Now, since $V\in\mc{V}$, for every node $\tilde{v}$, there are at
    most
    \begin{equation*}
        2 \pi (k+3) (\log(n)+5)\log(n)
        \leq K' k \log^2(n)
    \end{equation*}
    nodes at distance $[k-\log(n),k+1)$ for some constant $K' < \infty$. Hence
    the number of nodes in $S_2^k$ is at most
    \begin{equation}
        \label{eq:mimo5}
        \card{S_2^k}
        \leq K'k\log^2(n)\sum_{\tilde{k}=0}^{\log(n)}\card{S_2^{\tilde{k}}}.
    \end{equation}
    
    Combining \eqref{eq:mimo5} with \eqref{eq:mimo4} yields
    \begin{align}
        \label{eq:mimo6}
        C\Big(S_1, \cup_{k>\log(n)} S_2^k\Big) 
        & \leq \tilde{K}^2 \log^2(n) \sum_{k>\log(n)} \card{S_2^k} k^{4-\alpha} 
        \nonumber\\
        & \leq K' \tilde{K}^2 \log^4(n)
        \Big({\textstyle\sum_{\tilde{k}=0}^{\log(n)}}\card{S_2^{\tilde{k}}}\Big)
        \sum_{k>\log(n)} k^{5-\alpha} \nonumber\\
        & =  K'' \log^4(n)\sum_{\tilde{k}=0}^{\log(n)}\card{S_2^{\tilde{k}}},
    \end{align}
    for some constant $K'' < \infty$ depending only on $\alpha$, 
    and where we have used that $\alpha >6$.
    Finally, substituting \eqref{eq:mimo3} and \eqref{eq:mimo6} into
    \eqref{eq:mimo2} shows that
    \begin{align*}
        C(S_1,S_2)
        & \leq K_1\log^4(n)\sum_{k=0}^{\log(n)}\card{S_2^k} \\
        & = K_1\log^4(n)\big|\{v\in S^c: r_{S,v}< \log(n)+1\}\big|
    \end{align*}
    with
    \begin{equation*}
        K_1 \defeq K+K''.\qedhere
    \end{equation*}
\end{IEEEproof}

The next lemma shows that, for large path-loss exponents ($\alpha > 6$),
every cut is approximately achievable, i.e., for every cut there exists
an achievable unicast traffic matrix that has a sum rate across the cut that is
not much smaller than the cut capacity.
\begin{lemma}
    \label{thm:cuts}
    Under fast fading, for every $\alpha > 6$, there exists $b_5(n)\leq
    n^{o(1)}$ and $\lambdauc\in\Lambdauc(n)$  such that for any $n$, 
    $V(n)\in\mc{V}(n)$, and $S\subset V(n)$,
    \begin{equation}
        \label{eq:thm_cuts}
        C(S,S^c) 
        \leq b_5(n) \sum_{u\in S} \sum_{w\notin S} \lambdauc_{u,w}.
    \end{equation}
    Moreover, there exists a collection of channel gains $\mc{H}(n)$
    such that
    \begin{equation*}
        \Pp\big((h_{u,v})\in\mc{H}(n)\big)
        \geq 1-o(1)
    \end{equation*}
    as $n\to\infty$, and such that, for $(h_{u,v})\in\mc{H}(n)$,
    \eqref{eq:thm_cuts} holds for slow fading as well.
\end{lemma}
\begin{IEEEproof}
    By Lemma \ref{thm:mimo}, for $V\in\mc{V}$
    \begin{equation}
        \label{eq:cuts1}
        C(S,S^c) 
        \leq K_1 \log^4(n)\big|\{v\in S^c: r_{S,v}< \log(n)+1\}\big|.
    \end{equation}
    Construct a unicast traffic matrix $\lambdauc\in\Rp^{n\times
    n}$ as
    \begin{equation*}
        \lambdauc_{u,w} \defeq 
        \begin{cases}
            \kappa(n) & \text{if $r_{u,w}< \log(n)+1$}, \\
            0 & \text{otherwise},
        \end{cases}
    \end{equation*}
    for some function $\kappa(n)$. We now argue that for
    $\kappa(n)=\Theta(\log^{-3}(n))$ there exists $\tilde{b}(n)\geq
    n^{-o(1)}$ such that $\tilde{b}(n)\lambdauc\in\Lambdauc$.  This
    follows from \cite[Theorem 1]{nie2} (see also Section IX.C there),
    once we show that for every $\ell\in\{1,\ldots,
    L(n)\}$ and $i\in\{1,\ldots, 4^\ell\}$ we have
    \begin{subequations}
        \label{eq:nie2}
        \begin{align}
            \sum_{u\in V_{\ell,i}}\sum_{w\notin V_{\ell,i}} \lambdauc_{u,w}
            & \leq (4^{-\ell}n)^{2-\min\{3,\alpha\}/2}, \\
            \sum_{u\notin V_{\ell,i}}\sum_{w\in V_{\ell,i}} \lambdauc_{u,w}
            & \leq (4^{-\ell}n)^{2-\min\{3,\alpha\}/2},
        \end{align}
    \end{subequations}
    and, for all $w\in V$,
    \begin{align*}
        \sum_{u\neq w}\lambdauc_{u,w} & \leq 1, \\
        \sum_{u\neq w}\lambdauc_{w,u} & \leq 1.
    \end{align*}
    Since we assume that $V\in\mc{V}$, we have for all $w\in V$
    \begin{align*}
        \sum_{u\neq w}\lambdauc_{u,w} & \leq K \log^3(n)\kappa(n), \\
        \sum_{u\neq w}\lambdauc_{w,u} & \leq K \log^3(n)\kappa(n),
    \end{align*}
    for some constant $K<\infty$. By the locality of the unicast traffic
    matrix $\lambdauc$, it can be verified that this is sufficient for
    \eqref{eq:nie2} to hold with 
    \begin{equation*}
        \kappa(n) \defeq \frac{1}{K}\log^{-3}(n).
    \end{equation*}
    Hence \cite[Theorem 1]{nie2} applies, showing that 
    $\tilde{b}(n)\lambdauc\in\Lambdauc$ for fast fading, and the same
    holds for slow fading for $\mc{H}$ with 
    \begin{equation*}
        \Pp\big((h_{u,v})\in\mc{H}\big)
        \geq 1-o(1)
    \end{equation*}
    as $n\to\infty$. 

    Now, by construction of the unicast traffic matrix $\lambdauc$, 
    \begin{align*}
        \sum_{u\in S} \sum_{w\notin S} \lambdauc_{u,w} 
        & = \big|\{(u,w)\in S\times S^c: r_{u,w}<
        \log(n)+1\}\big|\kappa(n) \\
        & \geq \big|\{w\in S^c: r_{S,w}< \log(n)+1\}\big|\kappa(n).
    \end{align*}
    Combined with \eqref{eq:cuts1}, this implies that
    \begin{align*}
        C(S,S^c) 
        & \leq {K_1\log^4(n)}\big|\{w\in S^c: r_{S,w}< \log(n)+1\}\big|\\
        & \leq \frac{K_1\log^4(n)}{\kappa(n)}\sum_{u\in S}\sum_{w\notin S} \lambdauc_{u,w}.
    \end{align*}
    Since $\tilde{b}(n)\lambdauc\in\Lambdauc$, this proves the lemma
    with
    \begin{equation*}
        b_5(n)
        \defeq \frac{K_1\log^4(n)}{\kappa(n)\tilde{b}(n)}
        \leq n^{o(1)}.
    \end{equation*}
\end{IEEEproof}

We are now ready for the proof of Lemma \ref{thm:one}.

\begin{IEEEproof}[Proof of Lemma \ref{thm:one}]
    We wish to show that, for $\alpha > 6$, there exists $b_3(n) \leq
    n^{o(1)}$ such that 
    \begin{equation*}
        \rho(\lambdaca) \leq b_3(n) \hat{\rho}(\lambdaca)
    \end{equation*}
    with $\hat{\rho}(\lambdaca)$ as defined in
    \eqref{eq:rhohatdef}. Consider the traffic matrix
    $\rho(\lambda)\cdot\lambda$ and a cut $S\subset V$ in the wireless
    network. Assume we allow the nodes on each side of the cut to
    cooperate without any restriction---this can only increase
    achievable rates. The total amount of traffic that needs to be
    transmitted across the cut is at least 
    \begin{equation*}
        \rho(\lambda)
        \sum_{U\subset S}\sum_{w\notin S} \lambdaca_{U,w}.
    \end{equation*}
    The maximum achievable sum rate (with the aforementioned node
    cooperation) is given by $C(S,S^c)$, the MIMO capacity between the
    nodes in $S$ and in $S^c$. Therefore, 
    \begin{equation*}
        \rho(\lambdaca)
        \leq 
        \frac{C(S,S^c)}{\sum_{U\subset S}\sum_{w\notin S} \lambdaca_{U,w}}.
    \end{equation*}
    Since this is true for all cuts $S\subset V$, we may optimize over
    the choice of $S$ to obtain the bound
    \begin{equation}
        \label{eq:outer1}
        \rho(\lambdaca)
        \leq \min_{S\subset V} 
        \frac{C(S,S^c)}{\sum_{U\subset S}\sum_{w\notin S} \lambdaca_{U,w}}.
    \end{equation}

    We proceed by relating the cut $S$ in the wireless network to a cut
    $\tilde{S}$ in the graph $G$. By Lemma \ref{thm:cuts}, for $V\in\mc{V}$, 
    there exists $\lambdauc\in\Lambdauc$ such that for fast fading
    \begin{equation}
        \label{eq:outer2}
        C(S,S^c) 
        \leq b_5(n) \sum_{u\in S} \sum_{w\notin S} \lambdauc_{u,w},
    \end{equation}
    and \eqref{eq:outer2} holds also for slow fading if
    $(h_{u,v})\in\mc{H}$ with $\mc{H}$ defined as in Lemma
    \ref{thm:cuts}. By \cite[Theorem 1]{nie2} (see also the discussion
    in Section~IX.D there), for $\alpha > 6$ and $V\in\mc{V}$, there
    exists $K$ such that if $\lambdauc\in\Lambdauc$ then
    $K\log^{-6}(n)\lambdauc\in\Lambdauc_G$, where $G$ is the tree graph
    defined in Section \ref{sec:main_region}.

    Now, consider any $\tilde{S}\subset V_G$ such that
    $\tilde{S}\cap V = S$. Note that $\tilde{S}$ is a cut in $G$
    separating $S$ from $V \setminus S$.  Since
    $K\log^{-6}(n)\lambdauc\in\Lambdauc_G$,
    we thus have
    \begin{equation*}
        K\log^{-6}(n) \sum_{u\in S} \sum_{w\notin S} \lambdauc_{u,w}
        \leq \sum_{\substack{(u,v)\in E_G: \\ u\in \tilde{S},v\notin \tilde{S}}} c_{u,v}.
    \end{equation*}
    By minimizing over the choice of $\tilde{S}$ such that
    $\tilde{S}\cap V = S$, we obtain
    \begin{equation}
        \label{eq:outer3}
        K\log^{-6}(n) \sum_{u\in S} \sum_{w\notin S} \lambdauc_{u,w}
        \leq \min_{\tilde{S}: \tilde{S}\cap V = S}
        \sum_{\substack{(u,v)\in E_G: \\ u\in \tilde{S},v\notin \tilde{S}}} c_{u,v}.
    \end{equation}
    Combining \eqref{eq:outer2} and \eqref{eq:outer3} shows that
    \begin{equation*}
        C(S,S^c) 
        \leq \frac{b_5(n)}{K} \log^6(n) \min_{\tilde{S}: \tilde{S}\cap V = S}
        \sum_{\substack{(u,v)\in E_G: \\ u\in \tilde{S},v\notin \tilde{S}}} c_{u,v}.
    \end{equation*}

    Together with \eqref{eq:outer1}, and using Lemma~\ref{thm:mcv},
    this yields that with probability
    \begin{equation*}
        \Pp((h_{u,v})\in\mc{H}, V\in\mc{V}) \geq 1-o(1)
    \end{equation*}
    as $n\to\infty$, we have for any caching traffic matrix $\lambdaca$ 
    \begin{align*}
        \rho(\lambdaca)
        & \leq \min_{S\subset V} \frac{C(S,S^c)}{\sum_{U\subset S}\sum_{w\notin S} \lambdaca_{U,w}} \\
        & \leq b_3(n) \min_{S\subset V} \min_{\tilde{S}\in V_{G}: \tilde{S}\cap V = S}
        \frac{\sum_{\substack{(u,v)\in E_G: \\ u\in \tilde{S},v\notin \tilde{S}}} c_{u,v}}
        {\sum_{U\subset \tilde{S}\cap V}\sum_{w\in V\setminus\tilde{S}} \lambdaca_{U,w}} \\
        & = b_3(n) \min_{\tilde{S}\subset V_G} 
        \frac{\sum_{\substack{(u,v)\in E_G: \\ u\in \tilde{S},v\notin \tilde{S}}} c_{u,v}}
        {\sum_{U\subset \tilde{S}\cap V}\sum_{w\in V\setminus\tilde{S}} \lambdaca_{U,w}} \\
        & = b_3(n) \hat{\rho}(\lambdaca),
    \end{align*}
    with
    \begin{equation*}
        b_3(n) 
        \defeq \frac{b_5(n)}{K} \log^6(n)
        \leq n^{o(1)},
    \end{equation*}
    and where we have used \eqref{eq:rhohat} for the last equality.
\end{IEEEproof}

\subsection{Proof of Lemma \ref{thm:two}}
\label{sec:outer_two}

We wish to show that there exists $b_4(n)\geq n^{-o(1)}$ 
such that for any $\lambdaca$
\begin{equation}
    \label{eq:flows2}
    \phi(\lambdaca) \geq b_4(n)\hat{\rho}(\lambdaca)
\end{equation}
with $\hat{\rho}(\lambdaca)$ as defined in \eqref{eq:rhohatdef}. To this
end, we need to argue that whenever a caching traffic matrix can be
supported over the graph $G$, then there exists at least one cut in the
graph that is approximately saturated. In other words, we need to argue
that an approximate max-flow min-cut result holds for caching traffic
over $G$. 

The proof of the lemma proceeds as follows. We first transform the
\emph{undirected} graph $G$ into an \emph{directed} graph $\tilde{G}$
such that \emph{caching} traffic can be supported over $G$ if and only
if a corresponding \emph{unicast} traffic can be supported over
$\tilde{G}$. We then argue that for unicast traffic over $\tilde{G}$ an
approximate max-flow min-cut result holds. Finally,
we map this result for unicast traffic on $\tilde{G}$ back to $G$
to obtain the desired max-flow min-cut result for caching traffic over
$G$. 

Pick any $\lambdaca \in \Rp^{2^{n}\times n}$. For $\lambdaca=\bZero$,
$\phi(\lambda)$ and $\hat{\rho}(\lambda)$ are both infinite, and
the lemma trivially holds. Assume then that $\lambdaca\neq\bZero$.  By
rescaling $\lambdaca$ if required, we can then assume without loss of
generality that
\begin{equation}
    \label{eq:normalize}
    \sum_{(U,w)} \lambdaca_{U,w} = 1.
\end{equation}
Furthermore, we can assume that $\lambdaca_{U,w} = 0$ whenever $w\in U$,
since then $w$ already has access to the message it requests.

Recall that $G$ is an \emph{undirected} capacitated graph. We construct
a \emph{directed} capacitated graph
$\tilde{G}=(V_{\tilde{G}},E_{\tilde{G}})$ as illustrated in
Fig~\ref{fig:directed}.  Take the undirected graph $G$ and turn it into
a directed graph by splitting each edge $e\in E_G$ into two directed
edges each with the same capacity as $e$. Add $2^n$ additional nodes to
$V_G$, one for each subset $U\subset V(n)$. Connect the new node
$\tilde{u}$ corresponding to $U\subset V(n)$ to each node $u\in U$ by a
directed edge $(\tilde{u},u)$ with infinite capacity $c_{\tilde{u},u}=\infty$. 
\begin{figure}[tbp]
    \begin{center}
        \scalebox{0.8}{
        \input{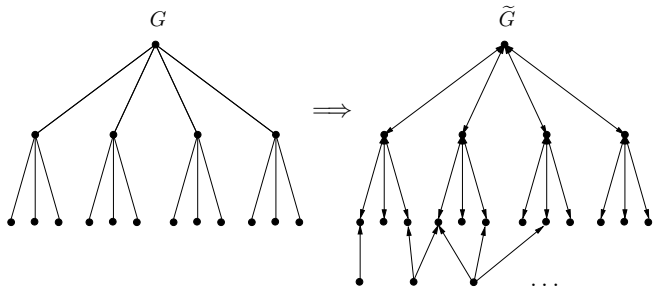}
        }
    \end{center}

    \caption{Construction of the directed graph $\tilde{G}$ from the undirected graph $G$.}

    \label{fig:directed}
\end{figure}

We call the directed version of $G$ that is contained in
$\tilde{G}$ as a subgraph its \emph{core}. Note that if some flows
can be routed through $G$ then the same flows can be routed through
the core of $\tilde{G}$, and if some flows can routed through
the core of $\tilde{G}$ then at least half of each flow can be
routed through $G$. Hence, for scaling purposes, the two are
equivalent. 

Now, assume we are given a caching traffic matrix $\lambdaca$ for $G$.
Construct a \emph{unicast} traffic matrix $\tlambdauc$ for $\tilde{G}$
by making for each $(U,w)$ pair in $G$ (i.e., $U\subset V(n)$, $w\in
V(n)$) the node $\tilde{u}$ in $\tilde{G}$ corresponding to $U$ a source
for $w$ with rate 
\begin{equation*}
    \tlambdauc_{\tilde{u},w} \defeq \lambdaca_{U,w}.
\end{equation*}
For all other node pairs, the traffic demand is set to zero.  Note that
with this construction, all traffic over $\tilde{G}$ originates at a
node in $V_{\tilde{G}}\setminus V_G $ of $\tilde{G}$ and is destined for
a node in $V(n)\subset V_{\tilde{G}}$. Denote by
$\Lambdauc_{\tilde{G}}(n)$ the set of such unicast traffic matrices that
are supportable on $\tilde{G}$, and define
\begin{equation*}
    \tilde{\phi}(\tlambdauc)
    \defeq \max\big\{\phi\geq 0: 
    \phi\tlambdauc \in \Lambdauc_{\tilde{G}}(n)\big\}
\end{equation*}
to be the equivalent description of $\Lambdauc_{\tilde{G}}(n)$.
By construction of $\tilde{G}$ from $G$, and by the above argument
relating $G$ to the core of $\tilde{G}$, we have
\begin{equation}
    \label{eq:inner2}
    \phi(\lambdaca)
    \geq \frac{1}{2}\tilde{\phi}(\tlambdauc).
\end{equation}
We have thus related caching traffic in the undirected graph $G$ to
unicast traffic in the directed graph $\tilde{G}$. 

We are then left with the problem of analyzing unicast traffic over
$\tilde{G}$. Recall that we have seen earlier that trivially
$\hat{\rho}(\lambdaca) \geq \phi(\lambdaca)$ (since the total flow over
each cut can be at most equal to the cut capacity). By
\eqref{eq:inner2}, this implies that $\tilde{\phi}(\tlambdauc) \leq
2\hat{\rho}(\lambdaca)$. The goal here is to establish that
$\hat{\rho}(\lambdaca)$ is also an approximate lower bound to
$\tilde{\phi}(\tlambdauc)$. This is nontrivial because it requires
showing that the polytope $\hLambdaca(n)$ with fewer constraints is
closely approximated by the polytope $\Lambdauc_{\tilde{G}}(n)$ with more
constraints. Specifically, we are looking for this approximation to be
of the form 
\begin{equation*}
    b(n)\hLambdaca(n) 
    \subset \Lambdauc_{\tilde{G}}(n)
    \subset 2\hLambdaca(n),
\end{equation*}
where $b(n) \geq n^{-o(1)}$.

In the recent literature on multicommodity flows, starting with works by
Leighton and Rao \cite{LR98}, and by Linial, London, and Rabinovich
\cite{llr}, such approximate max-flow min-cut results for unicast
traffic for undirected graphs have been studied. However, in our
context, two difficulties arise. First, $\tilde{G}$ is a directed graph.
While for undirected graphs with $m$ nodes $O(\log(m))$ approximation
results for the unicast capacity region of such graphs in terms of
cut-set bounds are known \cite{llr}, the best known approximation result
for general directed graphs is  $O(m^{11/23})$ up to polylog factors in
$m$ \cite{aga}. Second, the graph $\tilde{G}$ is exponentially big in
$n$.  More precisely, $\card{V_{\tilde{G}}}\geq 2^n$.  Hence even a
logarithmic (in the size $m$ of the graph) approximation result will
only yield a polynomial approximation in $n$. We are interested here in
an approximation ratio that scales like $n^{o(1)}$, i.e., strictly
sub-logarithmic in the size of $\card{V_{\tilde{G}}}$.  Nonetheless, as
we shall see, the special structure of $\tilde{G}$ can be
used to obtain an $O(\log(n)) \leq n^{o(1)}$ approximation for
$\Lambdauc_{\tilde{G}}(n)$ in terms of $\hLambdaca(n)$. 

We use an idea from \cite{kal}, namely that the unicast traffic problem
can be reduced to a maximum sum-rate problem.  More precisely, for a
subset $\tilde{F} \subset V_{\tilde{G}}\times V_{\tilde{G}}$
of $(u,w)$ pairs in $\tilde{G}$, define the \emph{maximum sum rate}
as
\begin{equation*}
    \tilde{\sigma}_{\tilde{F}}
    \defeq \max\big\{\tlambdauc_{\tilde{F}}:
    \tlambdauc\in \Lambdauc_{\tilde{G}}(n)\big\},
\end{equation*}
where here and in the following
\begin{equation*}
    \tlambdauc_{\tilde{F}} 
    \defeq \sum_{(u,w) \in \tilde{F}} \tlambdauc_{u,w}.
\end{equation*}
The quantity $\tilde{\sigma}_{\tilde{F}}$ is the largest sum rate that
can be supported between the source-destination pairs in $\tilde{F}$
over the graph $\tilde{G}$. 

We now argue that for every unicast traffic matrix $\tlambdauc$ there
exists $\tilde{F}$ such that the ratio
$\tilde{\sigma}_{\tilde{F}}/\tlambdauc_{\tilde{F}}$ is not too much
bigger than $\tilde{\phi}(\tlambdauc)$. 

\begin{lemma}
    \label{thm:three}
    Given $\tlambdauc$ on $\tilde{G}$ as described above, there
    exists a set $\tilde{F}$ of $(u,w)$ pairs with $u \in V_{\tilde{G}}\setminus V_G$
    and $w \in V(n) \subset V_{\tilde{G}}$ so that 
    \begin{equation}
        \label{eq:lemthree}
        \tilde{\phi}(\tlambdauc) 
        \geq \frac{1}{2(1+ \ln (2n^4))} 
        \frac{\tilde{\sigma}_{\tilde{F}}}{\tlambdauc_{\tilde{F}}}.
    \end{equation}
\end{lemma}
Recall that the unicast traffic matrix
$\tilde{\phi}(\tlambdauc)\tlambdauc$ is the largest scalar multiple of
$\tlambdauc$ that is supportable over $\tilde{G}$ by definition of
$\tilde{\phi}(\tlambdauc)$. Hence Lemma \ref{thm:three} shows that for a
point $\tilde{\phi}(\tlambdauc)\tlambdauc$ on the boundary of the region
$\Lambdauc_{\tilde{G}}(n)$ there exists a set of source-destination
pairs $\tilde{F}$ such that the total demand
$\tilde{\phi}(\tlambdauc)\tlambdauc_{\tilde{F}}$ between the pairs in $\tilde{F}$ is
almost as large as the maximum sum rate that is supportable between
$\tilde{F}$. Thus, for $\tilde{\phi}(\tlambdauc)\tlambdauc$, the pairs
in $\tilde{F}$ can be understood as the approximate bottleneck,
limiting further scaling of $\tlambdauc$ beyond the multiple
$\tilde{\phi}(\tlambdauc)$.

The next lemma links the ratio
$\tilde{\sigma}_{\tilde{F}}/\tlambdauc_{\tilde{F}}$ appearing in the
right-hand side of \eqref{eq:lemthree} to the equivalent description
$\hat{\rho}(\lambdaca)$ of the region $\hLambdaca(n)$.
\begin{lemma}
    \label{thm:four}
    For any set $\tilde{F}$ of $(u,w)$ pairs with
    $u\in V_{\tilde{G}}\setminus V_G$ and $w\in V(n) \subset V_{\tilde{G}}$,
    \begin{equation*}
        \frac{\tilde{\sigma}_{\tilde{F}}}{\tlambdauc_{\tilde{F}}} 
        \geq \frac{1}{4} \hat{\rho}(\lambdaca).
    \end{equation*}
\end{lemma}

Combining Lemmas \ref{thm:three} and \ref{thm:four} with \eqref{eq:inner2} 
shows that 
\begin{align*}
    \phi(\lambdaca) 
    & \geq \frac{1}{2} \tilde{\phi}(\tlambdauc) \\
    & \geq \frac{1}{4(1+\ln(2n^4))}
    \frac{\tilde{\sigma}_{\tilde{F}}}{\tlambdauc_{\tilde{F}}} \\
    & \geq \frac{1}{16(1+\ln(2n^4))} \hat{\rho}(\lambdaca).
\end{align*}
This establishes Lemma \ref{thm:two} with
\begin{equation*}
    b_4(n) \defeq \frac{1}{16(1+\ln(2n^4))} \geq n^{-o(1)}.
\end{equation*}
\IEEEQED

It remains to prove Lemmas \ref{thm:three} and \ref{thm:four}.

\begin{IEEEproof}[Proof of Lemma \ref{thm:three}]
    Given a unicast traffic matrix $\tlambdauc$ on $\tilde{G}$ as
    described above, we want to find a set of node pairs $\tilde{F}$
    such that $\tilde{\phi}(\tlambdauc)$ is not too much smaller than
    the ratio $\tilde{\sigma}_{\tilde{F}}/\tlambdauc_{\tilde{F}}$.
   
    First, note that $\tilde{\phi}(\tlambdauc)$ is the solution to the
    following linear program
    \begin{equation}
        \label{eq:primal1}
        \begin{array}{lr@{\,}ll}
            \text{max} & \phi & & \\
            \text{s.t.} & \sum_{p\in \tilde{P}_{u,w}} f_p 
            & \geq \phi \tlambdauc_{u,w} & \ \forall u,w\in V_{\tilde{G}}, \\
            & \sum_{p\in \tilde{P}: e\in p} f_p & \leq c_e & \ \forall e \in E_{\tilde{G}}, \\
            & f_p & \geq  0 & \ \forall p\in \tilde{P},
        \end{array}
    \end{equation}
    where $\tilde{P}_{u,w}$ is the collection of all paths in $\tilde{G}$ from
    node $u$ to node $w$, and 
    \begin{equation*}
        \tilde{P} 
        \defeq \bigcup_{(u,w)\in V_{\tilde{G}}\times V_{\tilde{G}}}\tilde{P}_{u,w}.
    \end{equation*}
    The corresponding dual linear program is
    \begin{equation}
        \label{eq:dual1}
        \begin{array}{lr@{\,}ll}
            \text{min} & \sum_{e\in E_{\tilde{G}}}c_e m_e & & \\
            \text{s.t.} & \sum_{e\in p} m_e & \geq d_{u,w} 
            & \ \forall u,w\in V_{\tilde{G}}, p\in \tilde{P}_{u,w}, \\
            & \sum_{u,w\in V_{\tilde{G}}} d_{u,w} \tlambdauc_{u,w} & \geq 1 & \\
            & m_e & \geq 0 & \ \forall e\in E_{\tilde{G}}, \\
            & d_{u,w} & \geq 0 & \ \forall u,w\in V_{\tilde{G}}.
        \end{array}
    \end{equation}
    Since the all-zero solution is feasible for the primal program
    \eqref{eq:primal1}, strong duality holds, i.e., the maximum in the
    primal \eqref{eq:primal1} is equal to the minimum in the dual
    \eqref{eq:dual1}. Moreover, by weak duality, any feasible solution
    to the dual problem \eqref{eq:dual1} yields an upper bound to the
    maximum in the primal \eqref{eq:primal1}.

    Second, $\tilde{\sigma}_{\tilde{F}}$ is the solution to the linear program
    \begin{equation}
        \label{eq:primal2}
        \begin{array}{lr@{\,}ll}
            \text{max} & \sum_{(u,w)\in \tilde{F}} \sum_{p\in \tilde{P}_{u,w}} f_p & & \\
            \text{s.t.} & \sum_{p\in \tilde{P}: e\in p} f_p 
            & \leq c_e & \ \forall e \in E_{\tilde{G}}, \\
            & f_p & \geq  0 & \ \forall p\in \tilde{P},
        \end{array}
    \end{equation}
    and its dual is
    \begin{equation}
        \label{eq:dual2}
        \begin{array}{lr@{\,}ll}
            \text{min} & \sum_{e\in E_{\tilde{G}}}c_e m_e & & \\
            \text{s.t.} & \sum_{e\in p} m_e 
            & \geq d_{u,w} & \ \forall u,w\in V_{\tilde{G}}, p\in \tilde{P}_{u,w}, \\
            & d_{u,w} & \geq 1 & \ \forall (u,w)\in \tilde{F}, \\
            & m_e & \geq 0 & \ \forall e\in E_{\tilde{G}}, \\
            & d_{u,w} & \geq 0 & \ \forall u,w\in V_{\tilde{G}}.
        \end{array}
    \end{equation}
    Again strong and weak duality hold.

    Let $(m_e^\star)_{e\in E_{\tilde{G}}}$, $(d_{u,w}^\star)_{u,w\in
    V_{\tilde{G}}}$ be a minimizer for the dual \eqref{eq:dual1} of the
    unicast traffic problem. By strong duality, the minimum of the dual
    \eqref{eq:dual1} is equal to the maximum of the corresponding primal
    \eqref{eq:primal1}. We now show how $(m_e^\star)$, $(d_{u,w}^\star)$ can
    be used to construct a feasible solution to the dual
    \eqref{eq:dual2} of the maximum sum-rate problem for a specific
    choice of subset $\tilde{F}$. By weak duality, this feasible
    solution for the dual \eqref{eq:dual2} yields an upper bound on the
    maximum in the corresponding primal \eqref{eq:primal2}.  This will
    allow us to lower bound $\tilde{\phi}(\tlambdauc)$ in terms of the
    ratio $\tilde{\sigma}_{\tilde{F}}/\tlambdauc_{\tilde{F}}$ as
    required.  

    Note first that we can assume without loss of optimality that 
    \begin{equation}
        \label{eq:dmin}
        d^\star_{u,w} = 
        \begin{cases}
            0 & \text{if $\tlambdauc_{u,w} = 0$,} \\
            \min_{p\in \tilde{P}_{u,w}} \sum_{e\in p} m_e^\star &
            \text{otherwise.}
        \end{cases}
    \end{equation}
    Now, since $c_e = \infty$ whenever $e\in E_{\tilde{G}}\setminus
    E_G$, we have $m_e^\star=0$ for those edges.  Since, in addition,
    $\tlambdauc_{u,w}>0$ only if $u\in V_{\tilde{G}}\setminus V_G$ and
    if $w$ is a leaf node of $G$, this implies that $(d^\star_{u,w})_{u,w\in
    V_{\tilde{G}}}$ can take at most $n^2$ different nonzero values,
    since there are at most that many distinct paths between leaf nodes
    in the tree graph $G$. Order these values in decreasing order
    \begin{equation*}
        d^\star_1 > d^\star_2 > \ldots > d^\star_K > d^\star_{K+1} = 0
    \end{equation*}
    with $K\leq n^2$, and define for $1\leq k\leq K$
    \begin{equation}
        \label{eq:kdef}
        \tlambdauc_k 
        \defeq \sum_{u,w\in V_{\tilde{G}}: d^\star_{u,w}=d^\star_k}\tlambdauc_{u,w}.
    \end{equation}

    We now argue that $d^\star_k \leq n^2$ for all $k\in\{1,\ldots,K\}$. In
    fact, assume $d^\star_1 > n^2$, then by \eqref{eq:dmin} there exists at
    least one edge $\tilde{e}$ such that $m_{\tilde{e}}^\star>n$, because in 
    any path $P_{u,w}$, there are at most $n$ edges with non-zero $m_{\tilde{e}}^\star$
    value. Hence
    \begin{equation*}
        \sum_{e\in E_{\tilde{G}}}c_e m_e^\star
        \geq c_{\tilde{e}} m_{\tilde{e}}^\star
        > n
    \end{equation*}
    since $c_{e}\geq 1$ for all $e\in E_{\tilde{G}}$. Due to strong
    duality, this implies that the solution of the linear program
    \eqref{eq:primal1}, i.e., the value of $\tilde{\phi}(\tlambdauc)$,
    is strictly larger than $n$. But that is not possible. Indeed, due
    to the normalization assumption \eqref{eq:normalize}, we have
    $\sum_{u,w\in V_{\tilde{G}}} \tlambdauc_{u,w} = 1$. 
    By construction, all destination nodes $w$ in $\tilde{V}_G$
    are in  $V \subset V_{\tilde{G}}$, and hence there are at most
    $n$ nodes $w$ with nonzero $\tlambdauc_{u,w}$. Together, this implies that
    for at least one node $w$ the total traffic into $w$
    satisfies 
    \begin{equation*}
        \sum_{u\in V_{\tilde{G}}} \tlambdauc_{u,w} \geq \frac{1}{n}.
    \end{equation*}
    By definition, $\tilde{\phi}(\tlambdauc) \tlambdauc$ must be
    supportable in $\tilde{G}$. Since $\tilde{\phi}(\tlambdauc) > 0$,
    and since, by assumption, $\tlambdauc_{U,w}=0$ whenever $w\in U$,
    this will induce a load strictly greater than one on the finite
    capacity edge incident on $w$. As $w\in V$, this
    edge has unit capacity, which contradicts that
    $\tilde{\phi}(\tlambdauc) \tlambdauc$ is supportable.  Therefore,
    $\tilde{\phi}(\tlambdauc)$ must be no more than $n$ and hence we
    obtain that $d^\star_k \leq d_1^\star \leq n^2$ for all $1\leq k\leq
    K$. 

    We now argue that at least one of $d_k^\star$ in $1\leq k\leq K$ 
    is not too small. To that end, let $k_1 < k_2 < \ldots < k_I$ 
    be such that
    \begin{equation}
        \label{eq:opti}
        \{k_i\}_{i=1}^I = \Big\{k: \tlambdauc_k \geq \frac{1}{2n^4}\Big\},
    \end{equation}
    with $\tlambdauc_k$ as defined in \eqref{eq:kdef}.
    Note that $I\geq 1$ since otherwise
    \begin{align*}
        \sum_{u,w\in V_{\tilde{G}}}\tlambdauc_{u,w}
        & = \sum_{k=1}^K \tlambdauc_{k} \\
        & <  \frac{K}{2n^4} \\
        & \leq 1,
    \end{align*}
    contradicting the normalization assumption \eqref{eq:normalize}.
    Define
    \begin{equation*}
        s_i \defeq \sum_{j=1}^i \tlambdauc_{k_j}.
    \end{equation*}
    Using that $(d_k^\star)$ is feasible for the dual \eqref{eq:dual1}, 
    that $d^\star_k \leq n^2$, and that $K\leq n^2$, we have
    \begin{align}
        \label{eq:dualbound}
        \sum_{i=1}^I d^\star_{k_i} \tlambdauc_{k_i}
        & \geq 1 - \sum_{k: \tlambdauc_k < 1/2n^4} d^\star_k \tlambdauc_k \nonumber\\
        & \geq 1 - K n^2\frac{1}{2n^4}  \nonumber\\
        & \geq \frac{1}{2}.
    \end{align}
    We argue that this implies existence of $i$ such that 
    \begin{equation}
        \label{eq:distbound}
        d_{k_i}^\star \geq \frac{1}{2s_i(1+\ln(2n^4))}.
    \end{equation}
    Indeed, assume \eqref{eq:distbound} is false for all $i$. Then
    \begin{align*}
        \sum_{i=1}^I d^\star_{k_i} \tlambdauc_{k_i} 
        & < \frac{1}{2(1+\ln(2n^4))}\sum_{i=1}^I 
        \frac{\tlambdauc_{k_i}}{s_i} \nonumber\\ 
        & \stackrel{(a)}{=}  \frac{1}{2(1+\ln(2n^4))}\Big(1+{\textstyle\sum_{i=2}^I} 
        \frac{s_i-s_{i-1}}{s_i}\Big) \\ 
        & \stackrel{(b)}{\leq}  \frac{1}{2(1+\ln(2n^4))}\Big(1+{\textstyle\sum_{i=2}^I}
        \big( \ln(s_i)-\ln(s_{i-1}) \big) \Big) \\
        & =  \frac{1}{2(1+\ln(2n^4))}
        \Big(1+\ln\Big({s_I}\big/{\tlambdauc_{k_1}}\Big)\Big) \nonumber\\
        & \stackrel{(c)}{\leq}  \frac{1}{2(1+\ln(2n^4))}\big(1+\ln(2n^4)\big) \\
        & = \frac{1}{2}, \nonumber
    \end{align*}
    where we have used that $I\geq 1$ in $(a)$, that
    $1-x\leq-\ln(x)$ for every $x\geq 0$ in $(b)$, and
    that $s_I\leq 1$ by \eqref{eq:normalize} and $\tlambdauc_{k_1} \geq
    \frac{1}{2n^4}$ in $(c)$.  This contradicts
    \eqref{eq:dualbound}, showing that \eqref{eq:distbound} must hold for some
    $i$. Consider this value of $i$ in the following.

    Now, consider the following set $\tilde{F}$ of $(u,w)$ pairs:
    \begin{equation*}
        \tilde{F} 
        \defeq \big\{ (u,w): d_{u,w}^\star \geq d_{k_i}^\star \big\}.
    \end{equation*}
    Note that, by \eqref{eq:dmin}, $\tilde{F}$ contains only pairs
    $(u,w)$ such that $u\in V_{\tilde{G}}\setminus V_G$ and $w\in
    V \subset V_{\tilde{G}}$ (i.e., nodes in $\tilde{G}$
    corresponding to leaf nodes in $G$). Set 
    \begin{align*}
        d_{u,w} & \defeq \frac{d_{u,w}^\star}{d^\star_{k_i}}, \\
        m_e & \defeq \frac{m_e^\star}{d^\star_{k_i}}
    \end{align*}
    for all $u,w\in V_{\tilde{G}}$.  Note that, for $(u,w)\in\tilde{F}$, 
    \begin{equation*}
        d_{u,w} 
        = \frac{d^\star_{u,w}}{d^\star_{k_i}} 
        \geq 1,
    \end{equation*}
    and that for all $u,w\in V_{\tilde{G}}$, $p\in \tilde{P}_{u,w}$,
    \begin{align*}
        \sum_{e\in p} m_e
        & = \frac{1}{d_{k_i}^\star}\sum_{e\in p} m_e^\star \\
        & \geq \frac{1}{d_{k_i}^\star} d_{u,w}^\star \\
        & =  d_{u,w},
    \end{align*}
    by feasibility of $(d_{u,w}^\star)$ and $(m_e^\star)$ for the dual
    \eqref{eq:dual1}. Hence, for this $\tilde{F}$, the choice of
    $(m_e)$ and $(d_{u,w})$ is feasible for the dual
    \eqref{eq:dual2}. By weak duality, any feasible solution for the
    dual \eqref{eq:dual2} yields an upper bound for the corresponding
    primal \eqref{eq:primal2}. Therefore
    \begin{align*}
        \tilde{\sigma}_{\tilde{F}} 
        & \leq \sum_{e\in E_{\tilde{G}}}c_e m_e \\
        & = \frac{1}{d^\star_{k_i}} \sum_{e\in E_{\tilde{G}}}c_e m_e^\star.
    \end{align*}
    By \eqref{eq:distbound},
    \begin{equation*}
        d_{k_i}^\star \geq \frac{1}{2s_i(1+\ln(2n^4))},
    \end{equation*}
    and, since $d_{k_j}^\star\geq d_{k_i}^\star$ for all $j\leq i$, 
    \begin{align*}
        s_i 
        & = \sum_{j=1}^i \tlambdauc_{k_j} \\
        & = \sum_{j=1}^i \sum_{(u,w): d_{u,w}^\star=d_{k_j}^\star} \tlambdauc_{u,w} \\
        & \leq \sum_{(u,w): d_{u,w}^\star\geq d_{k_i}^\star} \tlambdauc_{u,w} \\
        & = \sum_{(u,w)\in\tilde{F}} \tlambdauc_{u,w} \\
        & = \tlambdauc_{\tilde{F}}
    \end{align*}
    (note that the last equality is simply the definition of
    $\tlambdauc_{\tilde{F}}$). Therefore,
    \begin{align*}
        \tilde{\sigma}_{\tilde{F}} 
        & \leq \frac{1}{d^\star_{k_i}} \sum_{e\in E_{\tilde{G}}}c_e m_e^\star \\
        & \leq 2s_i(1+\ln(2n^4))\sum_{e\in E_{\tilde{G}}}c_e m_e^\star \\
        & \leq 2\tlambdauc_{\tilde{F}}(1+\ln(2n^4)) \sum_{e\in E_{\tilde{G}}}c_e m_e^\star.
    \end{align*}
    Since, by assumption, $(m_e^\star)$ is optimal for the dual
    \eqref{eq:dual1}, and by strong duality, we have
    \begin{equation*}
        \sum_{e\in E_{\tilde{G}}}c_e m_e^\star = \tilde{\phi}(\tlambdauc),
    \end{equation*}
    and hence
    \begin{equation*}
        \tilde{\phi}(\tlambdauc)
        \geq \frac{1}{2(1+\ln(2n^4))} 
        \frac{\tilde{\sigma}_{\tilde{F}}}{\tlambdauc_{\tilde{F}}}.\qedhere
    \end{equation*}
\end{IEEEproof}

\begin{IEEEproof}[Proof of Lemma \ref{thm:four}]
    We wish to analyze maximum sum rates
    $\tilde{\sigma}_{\tilde{F}}$ in $\tilde{G}$ for sets
    $\tilde{F}$ such that for $(\tu,w)\in \tilde{F}$ we have $\tu\in
    V_{\tilde{G}}\setminus V_G$ and $w\in V \subset V_G\subset
    V_{\tilde{G}}$. Notice that, due to this form of $\tilde{F}$
    and since the edges in $E_{\tilde{G}}\setminus E_G$ have 
    infinite capacity, this analysis can be done by considering only the
    core of $\tilde{G}$. More precisely, for a collection of node pairs
    $\tilde{F}$ in $\tilde{G}$ as above, we construct a collection
    of node pairs $F$ in $G$ as follows. For each
    $(\tilde{u},w)\in\tilde{F}$, note that by construction 
    $\tilde{u}$ is connected to a subset $U \subset V 
    \subset V_G \subset V_{\tilde{G}}$ of nodes. For each $(\tilde{u},w)\in\tilde{F}$, 
    add $(u,w)$ to $F$ for each such $u\in U$. Denote by $\sigma_F$ the maximum
    sum rate for $F$ in $G$. Since $G$ is the undirected version of the core
    of $\tilde{G}$, we have
    \begin{equation}
        \label{eq:inner4}
        \tilde{\sigma}_{\tilde{F}}
        \geq \sigma_F.
    \end{equation}

    For a collection of node pairs $F$ in $G$, we call a set of edges
    $M$ a \emph{multicut} for $F$ if in the graph $(V_G, E_G\setminus
    M)$ each pair in $F$ is disconnected. For a subset $M\subset E_G$,
    define
    \begin{equation*}
        c_M \defeq \sum_{e\in M} c_e.
    \end{equation*}
    From the definition of a multicut, it follows directly that $\sigma_F
    \leq c_M$. More surprisingly, it is shown in \cite[Theorem
    8]{gar} that if $G$ is an undirected tree, then for every $F\in
    V_G\times V_G$ there exists a multicut $M$ for $F$ such that
    \begin{equation}
        \label{eq:inner5}
        \sigma_F
        \geq \frac{1}{2}c_M.
    \end{equation}

    Next, we show how the edge cut $M\subset E_G$ can be transformed into a
    node cut $S\subset V_G$. Denote by $\{S_i\}$ the connected components of
    $(V_G,E_G\setminus M)$. We can assume without loss of generality that
    \begin{equation*}
        M = \bigcup_i (S_i\times S_i^c)\cap E_G,
    \end{equation*}
    since otherwise we can remove the additional edges from $M$ to create a
    smaller multicut for $F$. We therefore have
    \begin{equation}
        \label{eq:inner6a}
        c_M = \frac{1}{2}\sum_{i} c_{(S_i^c\times S_i)\cap E_G},
    \end{equation}
    since every edge in $M$ appears exactly twice in the sum on the
    right-hand side. Define for $S\subset V_G$ 
    \begin{equation*}
        \lambdaca_{S} 
        \defeq \sum_{U\subset S\cap V}\sum_{w\in V\setminus S} \lambdaca_{U,w},
    \end{equation*}
    as the total caching traffic that needs to be transmitted between
    $S\cap V$ and $V\setminus S$. $M$ is a multicut for $F$ induced by $\tilde{F}$,
    and hence for every $(\tu,w)\in \tilde{F}$ and the corresponding
    pair $(U,w)$, $M$ separates $w$ from all the nodes in $U$.
    Therefore, for each such $(U,w)$ pair, there exists a set $S_i$ such
    that $w \in S_i$, $U\subset S_i^c$. This shows that
    \begin{equation}
        \label{eq:inner6b}
        \tlambdauc_{\tilde{F}}
        \leq \sum_{i} \lambdaca_{S_i^c}.
    \end{equation}
    Equations \eqref{eq:inner5}, \eqref{eq:inner6a}, and
    \eqref{eq:inner6b} imply that there exists $j$ such that
    \begin{equation*}
        \begin{aligned}
            \frac{\tilde{\sigma}_{\tilde{F}}}{\tlambdauc_{\tilde{F}}}
            & \geq \frac{1}{4}
            \frac{\sum_i c_{(S_i^c\times S_i)\cap E_G}}{\sum_{i} \lambdaca_{S_i^c}} \\
            & \geq \frac{1}{4}
            \frac{c_{(S_{j}^c\times S_{j})\cap E_G}}{\lambdaca_{S_{j}^c}} \\
            & \geq \frac{1}{4}
            \min_{S\subset V_G} \frac{c_{(S\times S^c)\cap E_G}}{\lambdaca_{S}} \\
            & = \frac{1}{4}\hat{\rho}(\lambdaca),
        \end{aligned}
    \end{equation*}
    where in the last equality we have used \eqref{eq:rhohat}.
    This completes the proof of Lemma \ref{thm:four}. 
\end{IEEEproof}

\section{Conclusions}
\label{sec:conclusions}

We have analyzed the influence of caching on the performance of wireless
networks. Our approach is information-theoretic, yielding an inner bound
on the caching capacity region for all values $\alpha>2$ of path-loss
exponent, and a matching (in the scaling sense) outer bound for $\alpha
>6$. Thus, in the high path-loss regime $\alpha >6$, this provides a
scaling characterization of the complete caching capacity region.  Even
though this region is $2^{n}\times n$-dimensional, i.e., exponential in
the number of nodes $n$ in the wireless network, we have presented an
algorithm that checks approximate feasibility of a particular caching
traffic matrix efficiently, namely in polynomial time in the description
length of the caching traffic matrix. Achievability is proved using a
three-layer communication architecture. The three layers deal with
optimal selection of caches, choice of amount of necessary cooperation,
noise and interference, respectively. The matching (again in the scaling
sense) converse proves that addressing these questions separately is
without loss of order-optimality in the high path-loss regime. That is,
source-channel separation is close to optimal for caching traffic in
this regime.

We view this result as a step towards understanding the performance loss
incurred due to source-channel separation for the transmission of
arbitrarily correlated sources. Determining the performance loss for
such a separation based strategy for all values of $\alpha > 2$ for
caching traffic, and more generally for sources with arbitrary
correlation are interesting questions for future research.

\section{Acknowledgments}

The authors would like to thank the anonymous reviewers for their
comments.

\bibliographystyle{IEEEtran}
\bibliography{journal_abbr,caching}

\end{document}